%% file: main.tex
\documentclass[journal]{IEEEtran}
\IEEEoverridecommandlockouts

\usepackage{silence}
\usepackage{cite}

\usepackage{amsmath,amssymb,amsthm,fixmath}

\usepackage{color,colortbl}
\usepackage{xcolor}

\usepackage[inline]{enumitem}

\usepackage{siunitx}
\DeclareSIUnit{\dBm}{dBm}

\usepackage{graphicx}
\usepackage[labelformat=simple]{subcaption}

\usepackage{epstopdf}
\usepackage{lipsum}

\usepackage{cuted}
\usepackage{multirow,booktabs}
\usepackage{diagbox}
\usepackage{makecell}
\usepackage{hhline}
\usepackage{array} 
\newcolumntype{x}{!{\vrule width 2px}}
\newcolumntype{y}{!{\vrule width 1.5px}}

\usepackage{optidef}

\usepackage{nohyperref} 

\usepackage[shortcuts,acronym,automake]{glossaries}
\input{glossary}

\usepackage{tcolorbox}
\tcbuselibrary{many}
\newtheorem{theorem}{Theorem}
\newtheorem{corollary}{Corollary}
\newtheorem{proposition}{Proposition}
\theoremstyle{definition}
\newtheorem{remark}{Remark}

\newtcbtheorem[]{alg}{Algorithm}{fonttitle=\bfseries}{alg}
\usepackage[norelsize,linesnumbered,ruled]{algorithm2e}
\SetKwRepeat{Do}{do}{while}
\makeatletter
\newcommand{\removelatexerror} {\let\@latex@error\@gobble}
\makeatother

\newcommand{\superscript}[1]{^{\mathrm{#1}}}
\newcommand{\subscript}[1]{_{\mathrm{#1}}}
\newcommand{\diff}{\text{d}}

\usepackage[normalem]{ulem}

\usepackage[textsize=tiny,colorinlistoftodos]{todonotes}
\makeatletter
\define@key{todonotes}{bh}[]{
	\setkeys{todonotes}{author=\textbf{Bin}, color=lime!30}}%
\makeatother

\tikzstyle{note}=[rectangle, minimum width=3cm, draw = none, fill = none, minimum width = 1.5cm, anchor=center, align=left]
\tikzstyle{block}=[rectangle, draw, line width=1pt, fill = none, minimum width = 1cm, minimum height = 0.75cm, anchor=center, inner sep = 0.5mm, align=center]
\tikzstyle{arrow} = [thick,->,>=stealth]

\newif\ifreviewmode
\reviewmodetrue 

\ifreviewmode
\else
  \renewcommand{\todo}[1]{} 
\fi

\begin{document}
\let\url\gobble

\title{The Price of Ignorance: Information-Free Quotation for Data Retention in Machine Unlearning}
\author{
	Bin~Han,~\IEEEmembership{Senior Member,~IEEE,}
	Di~Feng,
	Zexin~Fang,~\IEEEmembership{Student Member,~IEEE,}\\
	Jie~Wang,~\IEEEmembership{Member,~IEEE,}
    and~Hans~D.~Schotten~\IEEEmembership{Member,~IEEE}
	\thanks{
		B. Han, Z. Fang, and H. D. Schotten are with RPTU University Kaiserslautern-Landau, Germany. D. Feng is with Dongbei University of Finance and Economics, China. J. Wang is with Tongji University, China. H. D. Schotten is with the German Research Center for Artificial Intelligence (DFKI), Germany.
    	B. Han (bin.han@rptu.de) and D. Feng (dfeng@dufe.edu.cn) are the corresponding authors.
		Part of this work has been accepted for publication at \emph{IEEE ICC 2026 Workshops}.

		\emph{AI disclosure:} In accordance with the IEEE policy on the use of AI, the authors disclose that Anthropic's Claude was employed in a strictly assistive capacity during the preparation of this manuscript---limited to language polishing of author-drafted text, auxiliary verification of the mathematical derivations prepared by the authors, and minor refactoring and debugging assistance of author-designed simulation code. All research ideas, the system model, the analytical framework, the theorems and their proofs, the algorithm design, the experimental methodology, and the interpretation of results were conceived, developed, and verified by the human authors, who take full and sole responsibility for the content and correctness of the paper.
	}
}

\maketitle

\begin{abstract}
When users exercise data deletion rights under the \ac{gdpr} and similar regulations, mobile network operators face a tradeoff: excessive machine unlearning degrades model accuracy and incurs retraining costs, yet existing pricing mechanisms for data retention require the server to know every user's private privacy and accuracy preferences, which is infeasible under the very regulations that motivate unlearning. We ask: \emph{what is the welfare cost of operating without this private information?} We design an information-free ascending quotation mechanism where the server broadcasts progressively higher prices and users self-select their data supply, requiring no knowledge of users' parameters. Under complete information, the protocol admits a unique subgame-perfect Nash equilibrium characterized by single-period selling. We formalize the \emph{Price of Ignorance}---the welfare gap between optimal personalized pricing (which knows everything) and our information-free quotation (which knows nothing)---and prove a three-regime efficiency ordering. Numerical evaluation across seven mechanisms and 5000 Monte Carlo runs shows that this price is near zero: the information-free mechanism achieves $\geqslant 99\%$ of the welfare of its information-intensive benchmarks, while providing noise-robust guarantees and comparable fairness.
\end{abstract}


\glsresetall

\section{Introduction}\label{sec:intro}
\Ac{llm} and \ac{genai} are increasingly deployed in next-generation mobile networks for intelligent network management and operations~\cite{BSA+2026survey}. 
\Acp{mno} deploy \acp{llm} for predictive maintenance, automated troubleshooting, traffic forecasting, and personalized service 
delivery. These applications require training on extensive user data--including cell handover records, packet inspection logs, geographic mobility traces, and application usage patterns--to achieve acceptable accuracy.

However, data collection raises privacy concerns: the data may contain sensitive information, and trained models may leak or misuse it against users' will. In network operations, this tension is especially acute because user data fed into \ac{llm} training directly influences \ac{qos} for the entire user base, coupling individual privacy with collective network performance.

Being aware of these privacy concerns, data regulations have been established in various countries and regions, represented by the \ac{gdpr} in the European Union~\cite{gdpr}, the \ac{ccpa} in California~\cite{ccpa}, and the \ac{pipl} in China~\cite{pipl}. Such regulations generally require the service providers to obtain the users' consent before collecting their data, and to provide the users with the right to delete their data afterwards.

To enable the removal of user data from \ac{llm} training pipelines--both from stored datasets and deployed models--\acp{mno} must implement machine unlearning~\cite{XZZ+2023machine}. Despite technical feasibility, excessive unlearning creates dual challenges 
for network operators: \begin{enumerate*}[label=(\arabic*)]
	\item computational overhead in retraining resource-
intensive \acp{llm}, and
	\item degraded predictive accuracy that cascades into worse 
traffic management, increased latency, and higher outage rates~\cite{CC2024price}.
\end{enumerate*}
A rigid unlearning policy thus harms both \ac{mno} operational efficiency and user experience, motivating market-based alternatives.

This calls for a data redemption framework that balances unlearning cost against privacy protection. A natural approach is data monetization: the \ac{llm} service provider offers users compensation to retain their data rather than redeem it. Users then choose freely between privacy and compensation based on their own preferences, while the provider benefits from reduced unlearning cost and better model accuracy.

The central challenge is to determine the price of data. While data pricing has been studied in the context of data marketplaces~\cite{ZBL2023survey}, existing works focus on active data trading for model training and largely ignore the redemption and unlearning scenario. This gap makes existing pricing models inapplicable to the data redemption context. For example, works such as~\cite{ZYZ2024price} ignore users' privacy loss when they sell data, while others such as~\cite{NZW+2020online}, though accounting for privacy, do not capture the computational cost of unlearning.

Cui and Cheung~\cite{CC2024price} first addressed pricing in data redemption, proposing an incentive mechanism for machine unlearning. Their model accounts for the server's unlearning cost (computation and accuracy degradation) and the users' privacy utility. A two-stage optimization is formulated: the server first determines an optimal unit price for user data, then each user decides how much data to sell.

The work in~\cite{CC2024price} has two limitations. First, the server's optimal pricing relies on distributional assumptions about users' privacy parameters that may not hold in practice. Second, the two-stage mechanism returns a single uniform price for all users and data batches; given heterogeneous privacy concerns and non-linear privacy utility, this can lead to suboptimal outcomes.

Both limitations share a common root: the server must know or estimate users' private parameters to set optimal prices. This raises a natural question: \emph{what is the welfare cost of operating without any such knowledge?} We call this the \textbf{Price of Ignorance}.

To answer it, we propose an iterative price discovery mechanism based on ascending quotations. The server (\ac{mno}) progressively raises the unit price for data retention, and users independently determine their supply at each quoted price---no user profiling, no parameter estimation. This enables flexible pricing that accommodates the non-linear cost and utility functions of both sides.

Our key contributions are:
\begin{enumerate}
	\item We extend the user payoff model of \cite{CC2025price} with \emph{accuracy-aware} users who incur disutility from model accuracy degradation, and a \emph{generalized privacy function} parameterized by elasticity $k_i$.
	\item We prove that under the ascending-price quotation, concentrating all sales in a single period is a strictly dominant strategy for each user (Theorem~\ref{thm:single_period}), and the resulting sequential game admits a unique \ac{spne} under complete information, characterized by backward induction (Theorem~\ref{thm:spne}).
	\item We formalize the \emph{Price of Ignorance}---the welfare gap between the information-intensive \ac{opp} and our information-free \ac{iiq}---and establish the ordering \ac{opp} $\geqslant$ \ac{ciq} $\geqslant$ \ac{iiq}, where the \ac{ciq}--\ac{iiq} gap depends on the oversupply-handling strategy (Proposition~\ref{prop:welfare_ordering}).
	\item We show that \ac{iiq} is robust to parameter estimation noise: unlike \ac{opp}, which requires the server to know every user's $(\lambda_i, k_i, \theta_i)$, \ac{iiq} provides noise-independent welfare guarantees.
	\item Numerical evaluation across seven mechanisms and 5000 Monte Carlo runs shows that the Price of Ignorance is near zero: \ac{iiq} achieves $\geqslant 99\%$ of \ac{opp} welfare, with comparable fairness across all mechanisms.
\end{enumerate}

The paper is organized as follows. Section~\ref{sec:related} reviews related work. Section~\ref{sec:model} presents the system model. Sections~\ref{sec:analyses}--\ref{sec:approach} analyze user incentives and propose the quotation mechanism. Section~\ref{sec:accuracy} extends the model to accuracy-aware users and derives the \ac{spne}. Section~\ref{sec:efficiency} compares the efficiency of three regimes. Section~\ref{sec:evaluation} provides numerical evaluation. Section~\ref{sec:discussion} discusses practical implications. Section~\ref{sec:conclusion} concludes.

\section{Related Work}\label{sec:related}

\subsection{Machine Unlearning}
Machine unlearning removes a data contributor's influence from a trained model without full retraining. The concept was first formalized by Cao and Yang~\cite{CY2015unlearning}, who proposed converting learning algorithms into summation forms to enable efficient data removal. Bourtoule et~al.~\cite{BRS+2021sisa} introduced the SISA framework, which partitions training data into shards to limit the retraining scope upon deletion requests. Subsequent works addressed unlearning in specific settings: Ginart et~al.~\cite{GGV+2019deletion} for $k$-means clustering, Brophy and Lowd~\cite{BL2021forests} for random forests, Neel et~al.~\cite{NRS2021descent} for gradient-based models via descent-to-delete, and Nguyen et~al.~\cite{NOD+2022markov} for Markov chain Monte Carlo based approaches. Gupta et~al.~\cite{GJN+2021adaptive} proposed adaptive unlearning with formal guarantees. Golatkar et~al.~\cite{GAS2020forgetting} introduced selective forgetting in deep networks using Fisher information, and Sekhari et~al.~\cite{SAK+2021remember} established theoretical bounds separating unlearning from differential privacy.

Recent surveys~\cite{XZZ+2023machine, NQX+2023unlearning} cover the breadth of unlearning techniques. Unlearning in \acp{llm} poses particular challenges, since the scale of models and training data makes full retraining prohibitively expensive~\cite{YFB+2024right}. Liu et~al.~\cite{LMZ+2024fedunlearn} extended unlearning to federated settings. On the adversarial side, Marchant et~al.~\cite{MRA2022forget} demonstrated poisoning attacks that increase unlearning cost, while Thudi et~al.~\cite{TJS+2022auditable} argued that auditable algorithmic definitions are necessary for meaningful unlearning guarantees. Yet the \emph{economic} side of unlearning---how data deletion rights affect the value proposition for data holders and model operators---remains largely open.

\subsection{Data Pricing and Valuation}
Data pricing in data marketplaces has received considerable attention~\cite{ZBL2023survey, Pei2022survey}: Fernandez et~al.~\cite{Fernandez2020market} developed data trading platforms, and Acemoglu et~al.~\cite{Acemoglu2022toomuch} analyzed pricing inefficiencies. Foundational approaches include query-based pricing~\cite{KPT2012query}, where buyers pay per query, and Shapley-value-based data valuation~\cite{GZ2019shapley, JGN+2019shapley, KZ2022shapley}, which distributes the value of a trained model among data contributors. Agarwal et~al.~\cite{ADG+2019marketplace} designed marketplace mechanisms with arbitrage-free pricing guarantees. Another line of work treats data as labor~\cite{AIJ+2018datalabor}, arguing that users should be compensated for their data contributions---a view that aligns with our redemption pricing framework.

However, existing data pricing works address \emph{forward} data trading---sellers supply data for model training. The \emph{reverse} scenario of data redemption, where data owners reclaim already-contributed data from deployed models, inverts the standard marketplace roles: the service provider holds the good (retained data) and the user exercises a deletion right. Cong et~al.~\cite{CMZ2022pipeline} studied pricing in ML pipelines but did not consider unlearning costs. Our work addresses this gap with pricing mechanisms for data redemption, where the server's cost function captures both accuracy degradation and retraining overhead.

\subsection{Incentive Mechanism Design}
Mechanism design for data-related decisions has been applied in federated learning and crowdsourcing. Kang et~al.~\cite{KXN+2019incentive} proposed a joint optimization approach for reliable federated learning with reputation-based incentives. Zhan et~al.~\cite{ZLZ+2020learning} designed learning-based incentive mechanisms for heterogeneous federated learning participants. Le et~al.~\cite{LTH+2021incentive} addressed incentives in wireless cellular networks. For privacy-aware settings, Jin et~al.~\cite{JWH+2019privacy} designed mechanisms for data crowdsourcing under differential privacy~\cite{Dwork2006differential} constraints, and Ghosh and Roth~\cite{GR2015selling} studied selling privacy at market. Acquisti et~al.~\cite{Acquisti2016economics} surveyed the economics of privacy, and Acquisti and Grossklags~\cite{AG2005privacy} showed that individuals' privacy decisions are bounded-rational, motivating mechanism designs that do not rely on users' strategic sophistication---a perspective that directly supports our \ac{iiq} mechanism.

Most closely related to our work is the line of research by Cui and Cheung~\cite{CC2024price, CC2025price}, who introduced the data redemption pricing problem and proposed a bilevel optimization framework. Their model captures both server cost (accuracy degradation and retraining time) and user privacy utility. However, their mechanism requires the server to know or estimate every user's privacy parameter $\lambda_i$ to compute optimal personalized prices. Our work relaxes this requirement by proposing an \emph{information-free} quotation mechanism where the server sets a uniform ascending price and users self-select, requiring no user profiling. A distinctive feature of our setting is that it involves \emph{divisible goods}---users choose continuous supply quantities rather than binary sell/keep decisions---and the server's demand is price-elastic through the accuracy loss function.

\section{System Model}\label{sec:model}
We model the interaction between an \ac{mno} deploying \acp{llm} for network optimization and its user base. The \ac{mno} (henceforth ``server'') operates an \ac{llm} trained on aggregated user data to perform tasks such as:
\begin{itemize}
	\item Traffic prediction: Forecasting bandwidth demand for proactive resource allocation;
	\item Anomaly detection: Identifying network faults or security threats;
	\item Service personalization: Recommending optimal connectivity plans.
\end{itemize}

Each user $i\in\mathcal{I}$ has contributed data $d_i$ (measured in normalized units, 
e.g., number of mobility records or session logs). Under \ac{gdpr} Article 17, 
users may request data deletion, triggering machine unlearning obligations. 
The following cost and utility models capture the economic tradeoffs.

\subsection{Server Cost Model}
Considering the cost model in \cite{CC2024price} that $C(x)=\alpha A(x)+ \beta T(x)$, where $x=\sum\limits_{i=1}^I x_i$, $x_i$ is the amount of data redemption for user $i\in\mathcal{I}$, $A(x)$ is the accuracy degradation of the model caused by unlearning, and $T(x)$ the computing time for executing the unlearning task. Note that $\alpha,A,\beta,T$ are all non-negative. Defining the total data amount for each user $i$ as $d_i$, and their sum $d=\sum\limits_{i=1}^Id_i$, according to \cite{CC2024price}:
\begin{align}
	A(x)&=A_1a^{A_2x}-A_3,\label{eq:A(x)}\\
	T(x)&=\begin{cases}
		0 & \text{if } x=0,\\
		T_0(d-x) & \text{if } x\in(0,d].
	\end{cases}\label{eq:T(x)}
\end{align}

From the server's perspective it is more convenient to focus on the amount of data to keep, i.e. to buy from the users, rather than the amount of data to redeem. So we define that $y_i=d_i-x_i$ and $y=d-x$, and rewrite the cost function as
\begin{align}
	C(y)&=\alpha A(y)+\beta T(y),\\
	A(y)&=A_1e^{A_2(d-y)}-A_3,\\
	T(y)&=\begin{cases}
		T_0y & \text{if } y\in[0,d),\\
		0 & \text{if } y=d.
	\end{cases}
\end{align}

\subsection{User Privacy Model}
Following \cite{CC2024price} we consider that for each user $i\in\mathcal{I}$ the utility of privacy is captured by
\begin{equation}
	U_i(x_i)=\lambda_i\ln(x_i+1),
\end{equation} 
where $\lambda_i$ is its privacy parameter. Similarly, focusing on the amount of data to keep on the server, we rewrite it As
\begin{equation}
	U_i(y_i)=\lambda_i\ln(d_i-y_i+1),
\end{equation}
which is concave and monotonically decreasing w.r.t. $y_i$.

\section{Incentive Analyses}\label{sec:analyses}
\subsection{Server's Incentive to Purchase Data}

The server has an incentive to keep as much data as to minimize the cost of unlearning. We denote this target amount of data to keep as $y\subscript{max}$. Note that $C(y)$ has one and only one jump discontinuity at $y=d$, for simplification of discussion we first focus on the continuous interval $y\in[0,d)$, and leave the case of $y=d$ for later discussion in Sec.~\ref{subsec:discontinuity}. Thus, $C(y)$ is in general a convex function over $y\in[0,d)$, while its monoticity, depending on the specific parameter settings, may have three diffferent cases, as shown in Fig.~\ref{fig:server_cost_cases}:
\begin{enumerate}[label=\arabic*):]
	\item monotonically increasing,
	\item monotonically decreasing, or
	\item non-monotone.
\end{enumerate}
\begin{figure}
	\centering
	\includegraphics[width=\linewidth]{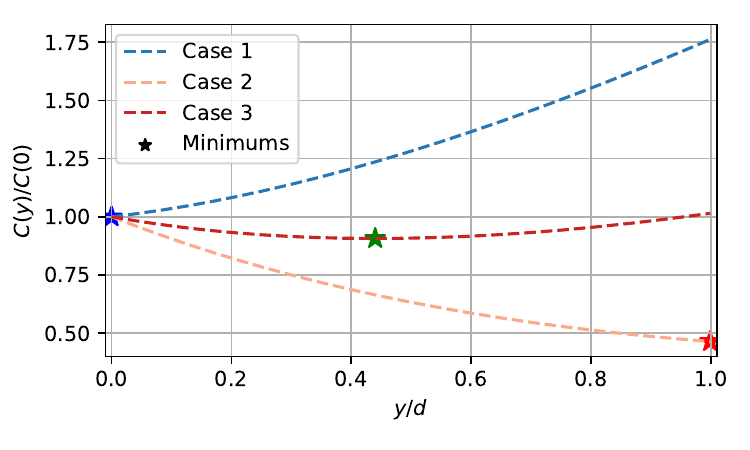}
	\caption{Different cases of the server's cost function.}
	\label{fig:server_cost_cases}
\end{figure}
In the first case, trivially, the server will not keep any data and $y\subscript{max}=0$. In the second case, the server tend to keep all data and $y\subscript{max}\to d^-$. In the third case, $y\subscript{max}$ can be obtained by looking for the zero derivative:
\begin{align}
	\left.\frac{\diff C(y)}{\diff y}\right\vert_{y\subscript{max}}&=(-\alpha A_1A_2\ln a)a^{A_2(d-y\subscript{max})}+\beta T_0=0,\\
	\Rightarrow y\subscript{max}&=d+\frac{1}{A_2}\log_a\left(\frac{\alpha A_1A_2\ln a}{\beta T_0}\right).\label{eq:ymax}
\end{align}

Now assume that the server has already managed to keep $y<y\subscript{max}$ amount of data, and offers to buy more data from users at a price $B$ per unit amount of data. The server's \emph{demand} $\eta$ will be a function of the condition $(y,B)$, which is:
\begin{align}
		&\eta(y,B)=\underset{\delta\in[0,y\subscript{max}-y]}{\arg\max}\left[C(y)-C(y+\delta)-B\delta\right]\label{eq:demand}\\
		=&\underset{\delta\in[0,y\subscript{max}-y]}{\arg\max}\left[\alpha A_1e^{A_2(d-y)}(1-e^{-A_2\delta)})-(\beta T_0+B)\delta\right].\nonumber
\end{align}

Especially, noticing that $C(y)$ is always monotonic over the interval $y\in[0,y\subscript{max})$, there is a unique maximum price $B\subscript{all}(y)$ to support a ``buy all'' strategy:
\begin{equation}
	\begin{split}
		&B\subscript{all}(y)=\frac{C(y\subscript{max})-C(y)}{y\subscript{max}-y}\\
		=&\frac{\alpha A_1e^{A_2d}\left(e^{-A_2y\subscript{max}}-e^{-A_2y}\right)}{y\subscript{max}-y}+\beta T_0.
	\end{split}
\end{equation}
With $y$ data already bought from the users, given any price $B\leqslant B\subscript{all}(y)$, the server will buy as much data from the users as possible, until it reaches the maximum amount $y\subscript{max}$ of data to keep. Note that in case of $y$-increasing cost $C$ (Case 1 in Fig.~\ref{fig:server_cost_cases}), $B\subscript{all}\leqslant 0$ and $y\subscript{max}=0$.

\subsection{User's Incentive to Sell Data}
Given the price $B$ offered by the server, each user $i$ can decide how much data it will sell to the server, based on how much data $y_i$ it has already sold previously. The change of privacy utility of selling $\delta_i$ more data is
\begin{equation}
	\begin{split}
		&\Delta U_i(y_i,\delta_i)=U_i(y_i+\delta_i) - U_i(y_i)\\
		=&\lambda_i\ln\left(1-\frac{\delta_i}{d_i-y_i+1}\right),
	\end{split}
\end{equation}
so the minimum price to convince user $i$ to sell $\delta_i$ more data must compensate this utility loss:
\begin{equation}
	\begin{split}
		&B_{\text{min},i}(y_i,\delta_i)=\frac{-\Delta U_i(y_i,\delta_i)}{\delta_i}\\
		=&\frac{\lambda_i}{\delta_i}\ln\left(1+\frac{\delta_i}{d_i-y_i+1-\delta_i}\right).
	\end{split}
\end{equation}
It is trivial that for all $y_i\in[0,d_i)$ and $\delta_i\geqslant 0$:
\begin{align}
	\frac{\partial}{\partial y_i}B_{\text{min},i}(y_i,\delta_i)&>0,\\
	\frac{\partial}{\partial \delta_i}B_{\text{min},i}(y_i,\delta_i)&>0,
\end{align}
so there is a minimal price for user $i$ to sell any amount of data in addition to the already-sold $y_i$:
\begin{equation}
	\lim\limits_{\delta_i\to0}B_{\text{min},i}(y_i,\delta_i)=\frac{\lambda_i}{d_i-y_i+1},
\end{equation}
and a minimal price for user $i$ to sell all the remaining data:
\begin{equation}
	B_{\text{all},i}(y_i)=B_{\text{min},i}(y_i,d_i-y_i)=\frac{\lambda_i}{d_i-y_i}\ln(d_i-y_i+1).  
\end{equation}

Moreover, given that user $i$ has already sold $y_i$ in priori, there is a maximal amount of data $q_i$ it is willing to sell to the server at price $p$, which we call $i$'s \emph{supply}:
\begin{equation}
	\begin{split}
		&q_i(y_i,B)=\underset{\delta\in[0,d_i-y_i]}{\arg\max}\left[p\delta+\Delta U_i(y_i,\delta)\right].\\
		=&\underset{\delta\in[0,d_i-y_i]}{\arg\max}\left[p\delta+\lambda_i\ln\left(1-\frac{\delta}{d_i-y_i+1}\right)\right].
	\end{split}   \label{equ:optimalselling}
\end{equation}

\section{Quotation-Based Price Discovery Mechanism}\label{sec:approach}
\subsection{Quotation-Based Price Discovery Protocol}
To balance the server's cost of unlearning and the users' privacy utility,
we propose an iterative price discovery mechanism through sequential
quotations. The protocol operates as follows:
\begin{enumerate}
	\item  At the beginning of each round $t$, the server updates its price $B^t$ for buying unit data from the users. We focus in this study on the ascending quotation format, where the data price is increasing w.r.t. time, i.e., for all $t<\tau$, $B^t<B^\tau$. It also estimates its demand for data $\eta^t=\eta(y^t, B^t)$, which is \emph{unknown} to the users. If the demand drops to zero, the quotation is terminated. Otherwise, the price is quoted to all users.\label{step:quote_price}
	\item Each user then independently decides how much data it can supply at this price.
	\item The server purchases the data from all users who are willing to sell at the quoted price. If the total supply exceeds its demand, it will only purchase its demanded amount of data. Otherwise, it will buy all data supplied.
	\item If the server has collected the target amount of data $y\subscript{max}$, the quotation is terminated. Otherwise, $t\gets t+1$, and go back to step \ref{step:quote_price}) for the next round.
\end{enumerate}

For the sake of low communication overhead and privacy protection, our approach is to let the users make decisions independently from each other. Thus, each user $i$ is considered \emph{blind} in the sense that at each period $t$, $i$ only knows the current data price $B^t$, its own privacy parameter $\lambda_i$, its own data amount $d_i$, and its own history of data selling $(y_i^\tau)_{\tau<t}$. As we are considering an ascending quotation, the price quoted by the server along the time assembles an ascending series $(B^\tau)_{1\leqslant\tau<\infty}$. 

\subsection{Users' Selling Strategy}
Noticing the $t$-monotonically increasing price $B^t$, when a user has incentive to sell data in round $t$, it has also the alternative option of keeping the data to sell it at higher price in a future round. Herewith we discuss the optimal selling strategy of users in such multi-round quotation.

Suppose the quotation is now at the beginning of round $t$. Each user $i$'s current remaining data that can be sold is $x_i^t\geqslant 0$. Given the information that $i$ knows, $i$ forms a belief $\mu_i^t\in [0,1]$, which represents a \emph{subjective} probability that the quotation is terminated after round $t$. Thus, $\mu_i^t$ is a function about server's remaining demand $\eta^t$ at $t$, and the supply $q_j^t$ of all users $j\in\mathcal{I}$. 

Moreover, at $t$, since $i$ only knows its own history $(y_i^\tau)_{\tau<t}$, and that $y_j^\tau$ is non-decreasing about $\tau$ for all $j\in\mathcal{I}$, $\mu_i^t$ is increasing with $i$'s own sold data $y_i^t$, and the collection $(\mu_i^\tau)_{1\leqslant \tau<\infty}$ is (point-wisely) decreasing with $\tau$, for any $(y_i^\tau)_{1\leqslant \tau<\infty}$.

Let $t_i$ be the first period that $i$ decides to sell, where according to the data price $B^{t_i}$, the optimal amount of selling for $i$ is 
automatically given by (\ref{equ:optimalselling}). Denote this amount of data be $q_i^{t_i}=q_i(0,B^{t_i})>0$. If $q_i^{t_i}=d_i$, it is straightfoward. Otherwise, supposing $q_i^{t_i}<d_i$, by the section of $t_i$ we have
\begin{equation}
	\mu_i^{t_i}(q_i^{t_i})\cdot V(q_i^{t_i},B^{t_i})\geqslant [1-\mu_i^{t_i}(0)] \cdot V(q_i^{t_i+1},B^{t_i+1}),\label{eq:better_sell_than_keep}
\end{equation} 
where $V(y,B)$ is $i$'s benefit by selling the amount of data $y$ with price $B$, according to his utility. Thus, Eq.~\eqref{eq:better_sell_than_keep} implies that the expected payoff of selling at $t_i$ is higher than the expected payoff of keeping data at $t_i$. In other words, as long as the users are blind about the server's demand and other users' trade record, they always perform the greedy strategy in selling data.

Furthermore, rearranging Eq.~\eqref{eq:better_sell_than_keep}, we have
\begin{equation}
	\frac{  \mu_i^{t_i} (q_i^{t_i})  }{[1-\mu_i^{t_i}(0)] }\geqslant \frac{V(q_i^{t_i+1},B^{t_i+1})}{V(q_i^{t_i},B^{t_i})}.	
\end{equation}
With the aforementioned monotonicity of $\mu_i$, we also have
\begin{equation}
	\frac{  \mu_i^{t_i+1} (q_i^{t_i+1})  }{[1-\mu_i^{t_i+1}(0)] }\geqslant \frac{  \mu_i^{t_i} (q_i^{t_i})  }{[1-\mu_i^{t_i}(0)] }.
\end{equation}
By the definitions of $q$ and $V$, $i$'s marginal benefit is decreasing, and hence
\begin{equation}
	\frac{V(q_i^{t_i+1},B^{t_i+1})}{V(q_i^{t_i},B^{t_i})}\geqslant \frac{V(q_i^{t_i+2},B^{t_i+2})}{V(q_i^{t_i+1},B^{t_i+1})}.	
\end{equation}
Thus, we can obtain that
\begin{equation}
	\frac{  \mu_i^{t_i+1} (q_i^{t_i+1})  }{[1-\mu_i^{t_i+1}(0)] }\geqslant  \frac{V(q_i^{t_i+2},B^{t_i+2})}{V(q_i^{t_i+1},B^{t_i+1})}.\label{eq:once_sell_never_back}
\end{equation}
Eq.~\eqref{eq:once_sell_never_back} implies, if $i$ starts to sell data at $t_i$, it will also choose to sell data in all future rounds, until having all its data sold.

\subsection{Bayesian Belief Model}\label{subsec:bayesian_belief}

The greedy selling strategy above relies on the belief $\mu_i^t$---the subjective probability that the quotation terminates after round $t$. We now formalize how a rational user forms and updates this belief through Bayesian learning.

\subsubsection{Why Termination Price, Not Demand}

Recall that at round $t$, user $i$ observes only the current price $B^t$ and whether the quotation is still active. Since the ascending price schedule is deterministic, knowing $B^t$ is equivalent to knowing~$t$. The only informative signal is the \emph{continuation event}: the fact that the server has not yet terminated the quotation.

A user cannot directly estimate the server's remaining demand $\eta^t$, as this depends on the server's private cost function $C(\cdot)$, the aggregate supply from all other users, and the total data endowment $d$---none of which is observable. However, the user \emph{can} form beliefs about the \emph{termination price} $B^T$---the price at which the server's demand is fully satisfied and the quotation ends---because:
\begin{itemize}
    \item At each round $t$ where the quotation remains active, the user learns that $B^T > B^t$, i.e., the termination price exceeds the current price.
    \item When the quotation terminates at some round $T$, the user observes $B^T = B^0 + T \cdot \Delta B$ exactly.
\end{itemize}
Thus, $B^T$ is a scalar quantity on which the user receives a sequence of increasingly informative signals, making Bayesian estimation both feasible and natural.

\subsubsection{Intra-Event Learning}\label{subsubsec:intra_event}

At the start of a redemption event, user $i$ holds a prior belief over the termination price:
\begin{equation}\label{eq:prior_BT}
    B^T \sim F_i^0,
\end{equation}
where $F_i^0$ is a cumulative distribution function on $[B^0, \infty)$. After observing that the quotation is still active at round $t$ (equivalently, $B^T > B^t$), user $i$ updates via Bayes' rule:
\begin{equation}\label{eq:posterior_BT}
    F_i^t(b) = \Pr(B^T \!\leqslant\! b \mid B^T \!>\! B^t) = \frac{F_i^0(b) - F_i^0(B^t)}{1 - F_i^0(B^t)},\ b > B^t.
\end{equation}
This is simply the left-truncation of the prior at $B^t$. The belief $\mu_i^t$ introduced above---the subjective probability that the quotation terminates after round $t$---is now given by
\begin{equation}\label{eq:mu_bayesian}
    \mu_i^t = \Pr(B^T \leqslant B^{t+1} \mid B^T > B^t) = \frac{f_i^0(B^t) \cdot \Delta B}{1 - F_i^0(B^t)} + o(\Delta B),
\end{equation}
where $f_i^0$ is the density of the prior. For small $\Delta B$, this is approximately the \emph{hazard rate} of the prior distribution evaluated at $B^t$\footnote{The hazard rate $h_i(b)$ is the reciprocal of the term $\frac{1-F(b)}{f(b)}$ that appears in Myerson's \emph{virtual value} $\phi(b) = b - 1/h(b)$ from optimal mechanism design~\cite{myerson1981optimal}. While structurally related, the two quantities play different roles: Myerson's virtual value is used by a mechanism designer to screen buyer types, whereas here the hazard rate drives the user's own selling-timing decision under subjective beliefs about the termination price.}:
\begin{equation}\label{eq:hazard_rate}
    \mu_i^t \approx h_i(B^t) \cdot \Delta B, \qquad h_i(b) \triangleq \frac{f_i^0(b)}{1 - F_i^0(b)}.
\end{equation}

\begin{remark}[Monotonicity and Hazard Rate]\label{rmk:monotonicity_DHR}
    The monotonicity of $\mu_i^t$ assumed above (decreasing over time for fixed selling history) corresponds to the prior $F_i^0$ having a \ac{dhr}. This holds for many common distributions, including Pareto, log-normal (in certain parameter ranges), and mixture distributions. Economically, DHR means that conditional on the quotation surviving to round $t$, the probability of termination in the next round decreases---reflecting the intuition that a quotation that has lasted long is likely to last longer.
\end{remark}

\subsubsection{Reinterpretation of Greedy Strategy}

The selling condition \eqref{eq:better_sell_than_keep} can be reinterpreted through the Bayesian lens:
\begin{equation}\label{eq:bayesian_sell_condition}
    h_i(B^{t_i}) \cdot V(q_i^{t_i}, B^{t_i}) \geqslant V(q_i^{t_i+1}, B^{t_i+1}) \cdot \left[1 - h_i(B^{t_i}) \cdot \Delta B\right],
\end{equation}
making the hazard rate $h_i(B^{t_i})$ the key determinant of selling timing. Under \ac{dhr} priors, the hazard rate decreases in $B^t$, so the greedy strategy (once-start-never-stop) from Eq.~\eqref{eq:once_sell_never_back} remains optimal. Under increasing hazard rate priors, users may marginally benefit from waiting, but the effect vanishes as $O(\Delta B)$ for small price increments.

\subsubsection{Inter-Event Learning}\label{subsubsec:inter_event}

When the server runs multiple redemption events over time (e.g., after model updates or new data collection cycles), users can refine their prior $F_i^0$ from past observations.

\textbf{Assumptions:} We require two conditions for cross-event learning to be meaningful:
\begin{enumerate}[label=(\alph*)]
    \item \emph{Stationary server type}: The server's cost structure $(A(\cdot), T(\cdot), \alpha, \beta)$ is approximately stable across events. This is reasonable when the server's model architecture and retraining infrastructure do not change drastically.
    \item \emph{Approximate population stability}: The user pool and its aggregate characteristics (distribution of $\lambda_i$, $\theta_i$, $d_i$) do not shift dramatically between events.
\end{enumerate}
Under~(a) and~(b), similar cost structure and user pool produce similar optimal data retention $y\subscript{max}$ and hence similar termination prices $B^T$ across events.

\textbf{Learning mechanism:} Suppose user $i$ has participated in $n$ past redemption events, observing realized termination prices $B^{T_1}, B^{T_2}, \ldots, B^{T_n}$. The user forms the prior for event $n+1$ by fitting $F_i^0$ to this sample. For instance, under a parametric model $B^T \sim \text{Gamma}(\alpha_0, \beta_0)$, standard conjugate Bayesian updating yields closed-form posteriors.

\textbf{Data accumulation dynamics:} Between redemption events, each user $i$ continues to generate new data that is incorporated into the server's model through normal operation. Let $g_i^{n}$ denote the amount of new data generated by user~$i$ during the interval between events $n$ and $n+1$. Then the data endowment at event $n+1$ is
\begin{equation}\label{eq:data_dynamics}
    d_i^{n+1} = y_i^n + g_i^{n},
\end{equation}
where $y_i^n$ is the data retained by the server from event $n$ (which remains in the model), and $g_i^{n}$ is the newly accumulated data. Data that was unlearned in event~$n$ (i.e., $d_i^n - y_i^n$) has been removed from the model and does not reappear unless the user re-consents.

This creates a dynamic linkage between events: a user who sells more data in event~$n$ (higher $y_i^n$) enters event $n+1$ with a larger data endowment, and hence potentially higher privacy stakes. Combined with the refined prior $F_i^0$ from the learning mechanism above, the user's selling threshold evolves across events---capturing both \emph{informational learning} (about $B^T$) and \emph{strategic positioning} (through data accumulation).

\begin{remark}[Event Frequency and Strategic Patience]\label{rmk:frequency}
    The frequency of redemption events does \emph{not} directly inform the user about the termination price $B^T$ within any given event. A server running quarterly events may reach the same $B^T$ as one running annually. However, event frequency affects the user's \emph{strategic patience} through two channels. First, frequent events reduce the temporal discounting of future selling revenue: the user can participate in a nearby future event rather than waiting a long time, lowering the opportunity cost of not selling now. Second, since each user's within-event supply~\eqref{equ:optimalselling} depends on the data endowment~$d_i$, and this endowment evolves across events via~\eqref{eq:data_dynamics}, more frequent events mean smaller per-interval data accumulation $g_i^n$, leading to a different supply profile at each subsequent event. Both effects influence supply behavior but are distinct from the belief about $B^T$.
\end{remark}

\subsection{Handling the Oversupply}\label{subsec:oversupply}
Due to the wide-sense $B$-monotonicity of $\eta$ (decreasing) and $q_i$ (increasing) for all $i\in\mathcal{I}$, when the quoted price $B^t$ is sufficiently high, it may occur the case of $\eta^t<\sum\limits_{i=1}^Iq_i^t$, i.e. the total supply from users exceeds the server's demand. In this case, the server shall allocate its demand $\eta$ to the different selling users.

While neither the unit price $B^t$ nor the demand $\eta^t$ is influenced by the allocation, the server's purchase is independent therefrom, either. However, since the different users are of different characteristics $(y_i, \lambda_i)$, the total payoff of users will be depending on the specific allocation of purchase. To perform an optimal allocation that maximizes the users' payoff, the server needs full knowledge about the privacy parameter $\lambda_i$ of each selling user $i$, which is assumed unavailable in our system model. Therefore, here we propose four privacy-knowledge-free strategies to handle such oversupply issue:
\begin{enumerate}
	\item \emph{Major sellers first}: The server will purchase data from users in a greedy manner and descending order upon their current supply $q_i^t$, until the demand is satisfied.
	\item \emph{Minor sellers first}: The server will purchase data from users in a greedy manner and ascending order upon their current supply $q_i^t$, until the demand is satisfied.
	\item \emph{Proportional}: The server will purchase data from all users, allocating its demand $\eta^t$ to every single selling user proportionally to their current supply $q_i^t$.
	\item \emph{Random order}: The server will purchase data from users in a greedy manner and random order.
\end{enumerate}

Remark that the oversupply will certainly trigger the satisfaction of server's optimal amount of data to keep $y\subscript{max}$, so it will only occur once, i.e., in the last round of the quotation.

\subsection{Handling the Discontinuity in Cost}\label{subsec:discontinuity}
It must be recalled that the server's cost function $C(y)$ has a jump discontinuity at $y=d$, which is so far not considered in the above analyses.

When $y\subscript{max}<d$, due to this jump discontinuity, the actual server's cost of unlearning all user data of amount $d$ is $C(d)=\alpha A(d)$, which can be significantly lower than $C(y\subscript{max})$. Thus, after achieving $y^t=y\subscript{max}$, the server may still have an incentive to see if it worth buying all the remaining data from the users. Nevertheless, since the server does not have full knowledge about every user's privacy parameter, it cannot simply estimate the optimal price to offer. It is neither rational to simply continue purchasing more data at prices higher than the marginal cost of unlearning. To address this issue, we propose a post-quotation procedure, where after buying $y\subscript{max}$ amount of data, the server will still keep announcing new prices to the users like if the quotation continues, \emph{but not purchasing}. Instead, it observes if all users are willing to sell all remaining data at the new price. Only if so, the server will purchase all remaining data at the last quoted price, and terminates the post-quotation step. Otherwise, no purchase will be made even if the supply is non-zero, and the server simply updates the price. This procedure continues until the price $B^t$ is raised so high that the cost to buy all remaining data exceeds the unlearning cost saved therewith, i.e. when $B^t\cdot\left(d-y\subscript{max}\right)>C\left(y\subscript{max}\right)-C(d)$.

\subsection{Algorithm Implementation}
Summarizing the above discussions, we propose an iterative quotation-based price discovery protocol in Algorithm~\ref{alg:quotation}. Especially, for practical implementation in real use scenarios, we consider discrete user data sets with unit size of $\Delta d$, so that all trades of user data can only be executed by amount of integer times of $\Delta d$.

\SetKwProg{Pn}{Function}{:}{}%
\SetKwFunction{FMain}{Main}
\begin{algorithm}[!htpb]
	\caption{Iterative Price Discovery Protocol}
	\label{alg:quotation}
	\scriptsize
	\DontPrintSemicolon
	\textbf{Input:} $a, A_1, A_2, A_3, T_0, \alpha, \beta, B^0, \Delta B, [d_1,d_2,\dots, d_I], \Delta d$,\;
	\textbf{Initialize:} $y^0\gets 0$, $t\gets 0$, $y\subscript{max}$ w.r.t. Eq.~\eqref{eq:ymax}, $\eta^0$ w.r.t. Eq.~\eqref{eq:demand}\;
	\While(\tcp*[f]{Quotation phase}){$\eta^t>0$}{
		Update $q_i^t$ for all $i\in\mathcal{I}$ w.r.t. Eq.~\eqref{equ:optimalselling}\;
		$q_i^t\gets \left\lfloor q_i^t/\Delta d\right\rfloor*\Delta d$\;
		\uIf{$\sum\limits_{i=1}^I q_i^t<\eta^t$}{
			$y_i^{t+1}\gets y_i^t+q_i^{t}$ for all $i\in\mathcal{I}$\;
		}
		\Else{
			Update $y_i^{t+1}$ for all $i\in\mathcal{I}$ upon the oversupply-handling strategy\;
		}
		$B^{t+1}\gets B^t+\Delta B$\;
		$\eta^{t+1}\gets \eta\left(y^{t+1},B^{t+1}\right)$\;
		$t\gets t+1$\;
	}
	\While(\tcp*[f]{Post-quotation phase}){$B^t\cdot\left(d-y^t\right)\leqslant C\left(y^t\right)-C(d)$}{
		Update $q_i^t$ for all $i\in\mathcal{I}$ w.r.t. Eq.~\eqref{equ:optimalselling}\;
		\If{$\sum\limits_{i=1}^I q_i^t=d-y$}{
			$y_i^{t+1}\gets d_i$ for all $i\in\mathcal{I}$\;
		}
		$B^{t+1}\gets B^t+\Delta B$\;
		$t\gets t+1$\;
	}
\end{algorithm}


\section{Extension to Accuracy-Aware Users}\label{sec:accuracy}

In Sections~\ref{sec:analyses}--\ref{sec:approach}, users' decisions are driven solely by privacy concerns, and each user acts independently without regard to other users' decisions. In practice, however, users also benefit from the server's model accuracy---for instance, an \ac{mno}'s \ac{llm} provides better traffic predictions and service quality when trained on more data. In this section, we extend the model following \cite{CC2025price} so that each user also accounts for the impact of aggregate data redemption on model accuracy.

We first analyze this extended model under complete information (where users can observe all preceding transactions) to establish a theoretical benchmark. Subsequently, in Sec.~\ref{sec:efficiency}, we compare this benchmark with the practical incomplete-information quotation mechanism of Sec.~\ref{sec:approach}, as well as with the optimal personalized pricing of \cite{CC2025price}.

\subsection{Extended User Model}\label{subsec:extended_model}

We generalize the user privacy function following \cite{CC2025price}. For each user $i\in\mathcal{I}$, the privacy utility of redeeming $x_i = d_i - y_i$ amount of data is
\begin{equation}\label{eq:general_privacy}
	P_i(x_i) = \begin{cases}
		\displaystyle\frac{\lambda_i(x_i+1)^{1-k_i}}{1-k_i}, & \text{if } k_i \in [0,1),\\[6pt]
		\lambda_i\ln(x_i+1), & \text{if } k_i = 1,
	\end{cases}
\end{equation}
where $\lambda_i \geqslant 0$ is user $i$'s privacy parameter and $k_i \in [0,1]$ is the privacy elasticity parameter. The model in Sec.~\ref{sec:model} corresponds to the special case $k_i = 1$ for all $i \in \mathcal{I}$.

Each user $i$ is also characterized by an \emph{accuracy parameter} $\theta_i \geqslant 0$ that captures $i$'s sensitivity to model accuracy degradation; a larger $\theta_i$ means a stronger preference for maintaining model quality.

Given a unit price $B$ offered by the server, user $i$'s payoff from selling $y_i$ amount of data to the server is
\begin{equation}\label{eq:accuracy_payoff}
	W_i(y_i, \mathbf{y}_{-i}, B) = P_i(d_i - y_i) + B \cdot y_i - \theta_i A\!\left(\textstyle\sum_{j=1}^I (d_j - y_j)\right),
\end{equation}
where $P_i(d_i - y_i)$ is the privacy utility from the redeemed portion, $B \cdot y_i$ is the compensation received for data retention, and $\theta_i A(\cdot)$ captures the disutility from accuracy degradation caused by the \emph{aggregate} data redemption across all users.

\begin{remark}
	Setting $\theta_i = 0$ and $k_i = 1$ for all $i \in \mathcal{I}$ recovers the privacy-only model analyzed in Sections~\ref{sec:analyses}--\ref{sec:approach}.
\end{remark}

\begin{remark}\label{rmk:cc2025_comparison}
	The payoff structure in \eqref{eq:accuracy_payoff} is equivalent to \cite[Eq.~(7)]{CC2025price}. The critical modeling difference is that in \cite{CC2025price}, the server sets \emph{personalized} prices $b_i$ for each user based on knowledge of $(\lambda_i, k_i, \theta_i)$, whereas in our framework, the server broadcasts a \emph{uniform} price $B$ without access to any user's private parameters.
\end{remark}

Note that when $\theta_i > 0$, user $i$'s payoff depends on the aggregate redemption of \emph{all} users through the accuracy term $\theta_i A(\cdot)$. This creates a strategic interdependence among users that does not exist in the privacy-only model. We analyze this interdependence under complete information in the remainder of this section.

\subsection{Complete-Information Benchmark}\label{subsec:complete_info}
To establish a theoretical benchmark, we first consider the setting where users have complete information: each user can observe all preceding transactions and anticipate subsequent users' responses. This allows us to characterize the game-theoretic equilibrium under ascending quotation prices. The welfare comparison between this complete-information quotation (CIQ) and the practical incomplete-information quotation (IIQ) is not straightforward---it depends on the oversupply-handling strategy and the accuracy externality structure, as we formally analyze in Sec.~\ref{sec:efficiency}.

Consider the quotation protocol of Sec.~\ref{sec:approach} applied to accuracy-aware users. Following Algorithm~\ref{alg:quotation}, the server quotes an ascending price sequence
\begin{equation}\label{eq:price_schedule}
	B^t = B^0 + t \cdot \Delta B, \qquad t = 0, 1, 2, \ldots,
\end{equation}
where $B^0 \geqslant 0$ is the initial price and $\Delta B > 0$ is the price increment. This schedule is deterministic and known to all participants. At each period $t$, each user $i$ decides the amount $y_i^t \geqslant 0$ to sell, subject to $\sum_{\tau=1}^{t} y_i^{\tau} \leqslant d_i$.

Suppose the quotation terminates at period $T$. Denoting $i$'s total data sold as $\bar{y}_i \triangleq \sum_{t=1}^{T} y_i^t$ and the system-wide total as $\bar{y} \triangleq \sum_{i=1}^{I} \bar{y}_i$, user $i$'s realized payoff is
\begin{equation}\label{eq:realized_payoff}
	W_i = P_i(d_i - \bar{y}_i) + \sum_{t=1}^{T} y_i^t \cdot B^t - \theta_i A\!\left(d - \bar{y}\right).
\end{equation}

\subsection{Single-Period Selling Property}\label{subsec:single_period}

The ascending price structure induces a strong structural property on users' optimal strategies.

\begin{theorem}[Single-Period Selling]\label{thm:single_period}
	Under an ascending price schedule $(B^t)_{t\geqslant 1}$, for any fixed strategy profile $\mathbf{y}_{-i}$ of all other users, it is a strictly dominant strategy for each user $i$ to sell data in at most one period. That is, there exists $t_i \in \{1, 2, \ldots, T\} \cup \{\infty\}$ such that $y_i^t = 0$ for all $t \neq t_i$.
\end{theorem}

\begin{proof}
	Fix user $i$ and any strategy profile of the other users, which determines the aggregate $\bar{y}_{-i} = \sum_{j \neq i} \bar{y}_j$. User $i$'s total data sold, $\bar{y}_i = \sum_t y_i^t$, uniquely determines both the privacy utility $P_i(d_i - \bar{y}_i)$ and the accuracy disutility $\theta_i A(d - \bar{y}_i - \bar{y}_{-i})$. Hence, the payoff \eqref{eq:realized_payoff} decomposes as
	\begin{equation}\label{eq:payoff_decomp}
		W_i = \underbrace{P_i(d_i - \bar{y}_i) - \theta_i A(d - \bar{y}_i - \bar{y}_{-i})}_{\triangleq\; \Phi_i(\bar{y}_i) \;\text{(timing-independent)}} + \underbrace{\sum_{t=1}^{T} y_i^t \cdot B^t}_{\text{(timing-dependent)}}.
	\end{equation}

	Now suppose, towards a contradiction, that $i$ sells positive amounts in two distinct periods $t < \tau$, i.e., $y_i^t > 0$ and $y_i^{\tau} > 0$. Since $B^t < B^{\tau}$,
	\begin{equation}
		y_i^t \cdot B^t + y_i^{\tau} \cdot B^{\tau} < (y_i^t + y_i^{\tau}) \cdot B^{\tau}.
	\end{equation}
	Shifting all sales from period $t$ to period $\tau$ strictly increases the timing-dependent component while leaving $\Phi_i(\bar{y}_i)$ unchanged (since $\bar{y}_i$ is preserved). Hence, selling across multiple periods is strictly dominated.
\end{proof}


\begin{remark}\label{rmk:comparison_privacy_only}
	The quotation protocol here is the same ascending-price mechanism as in Sec.~\ref{sec:approach}; the difference lies in the \emph{information structure}. In the incomplete-information setting of Sec.~\ref{sec:approach}, myopic users sell incrementally across rounds once they start (greedy strategy, cf.\ Eq.~\eqref{eq:once_sell_never_back}). Under the complete-information setting analyzed here, a stronger result holds: the ascending price structure implies that \emph{any} temporal splitting of sales is strictly dominated, so each user concentrates all sales in a single period.
\end{remark}

\subsection{Sequential Game and Equilibrium}\label{subsec:spne}

By Theorem~\ref{thm:single_period}, each active user $i$'s strategy in the complete-information game reduces to a pair $(t_i, \bar{y}_i)$: the period in which to sell and the total amount. Since the price schedule \eqref{eq:price_schedule} is deterministic, choosing a selling period $t_i$ is equivalent to choosing a selling price $B^{t_i} = B^0 + t_i \cdot \Delta B$. The period $t_i$ at which each user sells is \emph{endogenous}---it is determined by the equilibrium, not imposed a priori.

Without loss of generality, we exclude inactive users (who never sell) and reorder the remaining $I$ active users so that $t_1 \leqslant t_2 \leqslant \cdots \leqslant t_I$.

Under this ordering, the game becomes a finite sequential game of perfect information: user $1$ sells first at the lowest price, user $2$ at the next price, and so on, with each user observing all preceding transactions before making their decision.

\begin{remark}[Simultaneous Sellers]\label{rmk:simultaneous}
	If multiple users sell at the same period (i.e., $t_i = t_j$ for some $i \neq j$), they face the same price and form a simultaneous subgame within that round. In the practical mechanism, this is precisely the scenario handled by the oversupply strategies of Sec.~\ref{subsec:oversupply}: when aggregate supply exceeds demand, the server allocates according to one of the four proposed strategies. For the theoretical benchmark, we note that under continuous parameter distributions, exact ties in selling periods occur with probability zero. We therefore proceed under generic distinctness of selling periods, with the understanding that any ties are resolved by the existing oversupply mechanism.
\end{remark}

\begin{theorem}[Existence and Uniqueness of \acs{spne}]\label{thm:spne}
	The sequential game among accuracy-aware users under complete information and an ascending price schedule admits a unique pure-strategy \ac{spne}.
\end{theorem}

\begin{proof}
	\emph{Existence.} The game is a finite extensive-form game of perfect information. Each user's action space $[0, d_i]$ is compact, and the payoff functions \eqref{eq:realized_payoff} are continuous. By \cite{harris1985existence}, a pure-strategy \ac{spne} (in the sense of \cite{Selten1975spne}) exists.

	\emph{Uniqueness.} We proceed by backward induction. At each decision node, we show that the deciding user's payoff is strictly concave in their own action, guaranteeing a unique best response.

	Consider user $I$ (the last to sell). Given the aggregate predecessor supply $s_{I-1} = \sum_{i < I} \bar{y}_i$, user $I$'s payoff is
	\begin{equation}\label{eq:last_user_payoff}
		V_I(\bar{y}_I;\, s_{I-1}) = P_I(d_I - \bar{y}_I) - \theta_I A(d - \bar{y}_I - s_{I-1}) + \bar{y}_I \cdot B^{t_I}.
	\end{equation}
	The second derivative with respect to $\bar{y}_I$ is
	\begin{equation}\label{eq:concavity_check}
		\frac{\partial^2 V_I}{\partial \bar{y}_I^2} = P_I''(d_I - \bar{y}_I) - \theta_I A''(d - \bar{y}_I - s_{I-1}).
	\end{equation}
	Since $P_I$ is concave ($P_I'' \leqslant 0$) and $A$ is convex ($A'' \geqslant 0$) with $\theta_I \geqslant 0$, we have $\frac{\partial^2 V_I}{\partial \bar{y}_I^2} \leqslant 0$. Strict concavity holds whenever $P_I''(d_I - \bar{y}_I) < 0$ (which is the case for the privacy functions in \eqref{eq:general_privacy} for all $k_i \in [0,1]$ and $\bar{y}_I < d_I$) or $\theta_I > 0$ and $A'' > 0$ (which holds for the exponential accuracy model \eqref{eq:A(x)}).

	Thus, $V_I$ is strictly concave in $\bar{y}_I$ over $[0, d_I]$, implying a unique maximizer $\bar{y}_I^*(s_{I-1})$.

	By substituting $\bar{y}_I^*(s_{I-1})$ into user $(I\!-\!1)$'s problem, the same argument applies: user $(I\!-\!1)$'s payoff, with user $I$'s response embedded, remains strictly concave in $\bar{y}_{I-1}$, yielding a unique best response $\bar{y}_{I-1}^*(s_{I-2})$. Proceeding inductively to user $1$, each step yields a unique best response. Since the backward induction is single-valued at every node, the resulting \ac{spne} is unique.
\end{proof}

\subsection{Backward Induction Characterization}\label{subsec:backward_induction}

The unique \ac{spne} of Theorem~\ref{thm:spne} is characterized constructively as follows.

\textbf{Step $I$ (Last user):} Given aggregate predecessor supply $s_{I-1} \triangleq \sum_{i=1}^{I-1} \bar{y}_i$, user $I$'s optimal supply is
\begin{equation}\label{eq:bi_step_I}
	\bar{y}_I^*(s_{I-1}) = \underset{\bar{y}_I \in [0, d_I]}{\arg\max}\; V_I(\bar{y}_I;\, s_{I-1}),
\end{equation}
which can be obtained from the \ac{foc}
\begin{equation}\label{eq:foc_I}
	-P_I'(d_I - \bar{y}_I) + \theta_I A'(d - \bar{y}_I - s_{I-1}) + B^{t_I} = 0,
\end{equation}
subject to the constraint $\bar{y}_I \in [0, d_I]$.

For the log-privacy case ($k_I = 1$) and exponential accuracy model, the \ac{foc} \eqref{eq:foc_I} specializes to
\begin{equation}\label{eq:foc_I_explicit}
	\frac{\lambda_I}{d_I - \bar{y}_I + 1} + \theta_I A_1 A_2 \ln(a) \cdot a^{A_2(d - \bar{y}_I - s_{I-1})} = B^{t_I},
\end{equation}
where the left-hand side is strictly decreasing in $\bar{y}_I$, confirming uniqueness.

\textbf{Step $i$ (General user, $i = I\!-\!1, I\!-\!2, \ldots, 1$):} Given $s_{i-1} = \sum_{j=1}^{i-1} \bar{y}_j$ and the uniquely determined responses $\bar{y}_{i+1}^*, \ldots, \bar{y}_I^*$ of all subsequent users (which depend on $s_{i-1}$ and $\bar{y}_i$ through the backward induction), user $i$ solves
\begin{multline}\label{eq:bi_step_i}
	\bar{y}_i^*(s_{i-1}) = \underset{\bar{y}_i \in [0, d_i]}{\arg\max}\; P_i(d_i - \bar{y}_i) + \bar{y}_i \cdot B^{t_i}\\
	 - \theta_i A\!\big(d - s_{i-1} - \bar{y}_i - R_i(s_{i-1}, \bar{y}_i)\big),
\end{multline}
where $R_i(s_{i-1}, \bar{y}_i) \triangleq \sum_{j=i+1}^{I} \bar{y}_j^*$ is the aggregate response of all users after $i$, recursively determined by the backward induction.

The \ac{spne} outcome is the profile $(\bar{y}_1^*, \bar{y}_2^*, \ldots, \bar{y}_I^*)$ obtained by evaluating the chain $s_0^* = 0$, $\bar{y}_1^* = \bar{y}_1^*(s_0^*)$, $s_1^* = \bar{y}_1^*$, $\bar{y}_2^* = \bar{y}_2^*(s_1^*)$, and so on.

\begin{corollary}[Reduction of Server's Problem]\label{cor:pricing}
	For the linear price schedule \eqref{eq:price_schedule}, the unique \ac{spne} induces a unique total data supply $\bar{y}^* = \sum_{i=1}^I \bar{y}_i^*$ that depends on the parameters $(B^0, \Delta B)$. Since each user sells at exactly one period, the server's cost minimization over the price schedule reduces to choosing $(B^0, \Delta B)$:
	\begin{equation}\label{eq:server_opt}
		\min_{B^0 \geqslant 0,\; \Delta B > 0} \left[\sum_{i=1}^{I} B^{t_i^*} \cdot \bar{y}_i^*(B^0, \Delta B) + C\!\left(\sum_{i=1}^I \bar{y}_i^*\right)\right],
	\end{equation}
	where $t_i^* = t_i^*(B^0, \Delta B)$ is the endogenous selling period of user $i$ in the \ac{spne} and $B^{t_i^*} = B^0 + t_i^* \cdot \Delta B$.
\end{corollary}

\subsection{Information-Free Implementation}\label{subsec:info_free}

A key practical advantage of the quotation mechanism is that it does not require the server to solve Problem \eqref{eq:server_opt} or know user parameters.

\begin{proposition}[Information-Free Implementation]\label{prop:info_free}
	The quotation-based price discovery mechanism (Algorithm~\ref{alg:quotation}) implements a data retention outcome without requiring the server to know any user's private parameters $(\lambda_i, k_i, \theta_i)$. At each round, the server only needs to:
	\begin{enumerate}
		\item broadcast the current uniform price $B^t$, and
		\item observe the aggregate supply $\sum_i y_i^t$.
	\end{enumerate}
	Each user independently determines their supply based on locally known parameters and the quoted price.
\end{proposition}

In the practical quotation mechanism, users do \emph{not} observe other users' individual trades and cannot perform backward induction. The complete-information \ac{spne} of Theorem~\ref{thm:spne} therefore serves as a theoretical benchmark: it represents the best outcome achievable under the ascending quotation format if users had full strategic foresight. The actual outcome under incomplete information---where users follow myopic best-response strategies as in Sec.~\ref{sec:approach}---may differ, and we quantify this gap in Sec.~\ref{sec:efficiency}.


\section{Efficiency Analysis}\label{sec:efficiency}

We now formalize the \emph{Price of Ignorance}: the welfare cost of operating the quotation mechanism without private user information. We compare three regimes, ranging from the theoretical optimum to the practical mechanism:

\begin{enumerate}
	\item \textbf{\acf{opp}:} The server knows all user parameters and sets personalized prices $b_i$ for each user, as in \cite{CC2025price}. This is the social-welfare--maximizing benchmark under complete information.
	\item \textbf{\acf{ciq}:} The server runs the ascending quotation protocol, and users play the \ac{spne} of Theorem~\ref{thm:spne} with full strategic foresight. Prices are uniform but ascending.
	\item \textbf{\acf{iiq}:} The practical mechanism of Sec.~\ref{sec:approach}, extended to accuracy-aware users. Users follow myopic best-response strategies without observing others' trades.
\end{enumerate}

Let $\xi\subscript{OPP}$, $\xi\subscript{CIQ}$, and $\xi\subscript{IIQ}$ denote the social welfare achieved under each regime. The first ordering is immediate:
\begin{equation}\label{eq:welfare_ordering_opp_ciq}
	\xi\subscript{OPP} \geqslant \xi\subscript{CIQ},
\end{equation}
since personalized pricing can always replicate a uniform price (by setting $b_i = B$ for all $i$).

The comparison between \ac{ciq} and \ac{iiq}, however, is more nuanced and depends on the oversupply-handling strategy, as we analyze below.

\subsection{Price of Uniformity}

The gap $\xi\subscript{OPP} - \xi\subscript{CIQ}$ quantifies the cost of using uniform prices instead of personalized ones. This gap arises because the ascending quotation offers a single price at each round, whereas \ac{opp} can tailor prices to individual user characteristics.

\subsection{Price of Ignorance}\label{subsec:price_of_ignorance}

The welfare gap between \ac{ciq} and \ac{iiq} arises from differences in both the selling structure (single-period vs.\ incremental) and the strategic treatment of the accuracy externality. Myopic users in \ac{iiq} ignore the accuracy externality and may offer to sell more data than under the strategic play of \ac{ciq}. However, the server's demand $\eta$ is determined by its own cost function and is independent of users' strategies. The server purchases data only up to its demand, so the total amount of data actually retained---and hence the aggregate accuracy term $A(d - \bar{y})$---is the same across \ac{ciq} and \ac{iiq} for a given price path.

The welfare difference between \ac{ciq} and \ac{iiq} therefore reduces to a \emph{misallocation effect}: in \ac{ciq}, strategic users self-select efficiently into selling, whereas in \ac{iiq}, myopic oversupply in the final round is resolved by the oversupply-handling strategy (Sec.~\ref{subsec:oversupply}). Specifically:
\begin{itemize}
	\item The aggregate data retained by the server, $\bar{y}$, and the server's cost are identical.
	\item The accuracy disutility $\sum_i \theta_i A(d - \bar{y})$ is identical (it depends only on total $\bar{y}$).
	\item The difference lies in \emph{which users are partially cut} in the oversupply round, affecting: (a)~the distribution of compensation across users, and (b)~the distribution of privacy utilities.
\end{itemize}

Since transfers (compensation) cancel out in social welfare, the net welfare effect depends solely on how the oversupply-handling strategy allocates the cut among users, which determines the aggregate privacy utility. This yields the following result.

\begin{proposition}[Oversupply-Dependent Welfare Ordering]\label{prop:welfare_ordering}
	Under accuracy-aware users, the welfare comparison between \ac{ciq} and \ac{iiq} depends on the oversupply-handling strategy:
	\begin{equation}\label{eq:welfare_gap}
		\xi\subscript{CIQ} - \xi\subscript{IIQ} = \sum_{i \in \mathcal{I}} \left[P_i(d_i - \bar{y}_i\superscript{CIQ}) - P_i(d_i - \bar{y}_i\superscript{IIQ})\right],
	\end{equation}
	where $\bar{y}_i\superscript{CIQ}$ and $\bar{y}_i\superscript{IIQ}$ denote user $i$'s data sold under the respective regimes. The per-user allocations $\bar{y}_i\superscript{IIQ}$ are determined by the oversupply-handling strategy of Sec.~\ref{subsec:oversupply}, subject to $\sum_i \bar{y}_i\superscript{IIQ} = \sum_i \bar{y}_i\superscript{CIQ}$.
\end{proposition}

\begin{remark}
	This result implies that the choice of oversupply-handling strategy has welfare consequences beyond the privacy-only model. We evaluate this numerically in Sec.~\ref{sec:evaluation}, comparing all four strategies under accuracy-aware users to identify which minimizes the \ac{ciq}--\ac{iiq} gap.
\end{remark}

\subsection{Robustness Advantage}\label{subsec:robustness}

\Ac{iiq} has a further advantage over \ac{opp}: \emph{robustness to information errors}. The \ac{opp} mechanism of \cite{CC2025price} requires the server to know each user's private parameters $(\lambda_i, k_i, \theta_i)$ to solve the bilevel optimization \cite[Problem~(11)]{CC2025price}. In practice, these parameters must be estimated, introducing errors.

We model the server's estimation error as follows. The server observes noisy estimates
\begin{equation}\label{eq:noisy_params}
	\hat{\lambda}_i = \lambda_i(1 + \varepsilon_i^{\lambda}), \qquad \hat{\theta}_i = \theta_i(1 + \varepsilon_i^{\theta}),
\end{equation}
where $\varepsilon_i^{\lambda}, \varepsilon_i^{\theta} \overset{\text{i.i.d.}}{\sim} \mathcal{N}(0, \sigma^2)$ are independent multiplicative noise terms, and $\sigma \geqslant 0$ controls the estimation quality. We use multiplicative noise because (i)~it preserves the non-negativity of $\lambda_i$ and $\theta_i$ (for small $\sigma$), and (ii)~it is scale-invariant, reflecting that estimation difficulty is proportional to parameter magnitude. The elasticity parameter $k_i \in [0,1]$ is assumed known, as it characterizes a structural property of the privacy function that can be estimated from aggregate behavior.

Given noisy estimates $(\hat{\lambda}_i, k_i, \hat{\theta}_i)$, the server solves \ac{opp}'s bilevel optimization \cite[Problem~(11)]{CC2025price} and obtains suboptimal personalized prices $\hat{\mathbf{b}}(\sigma)$. Users then respond to $\hat{\mathbf{b}}(\sigma)$ using their \emph{true} parameters, reaching a Nash equilibrium $\mathbf{x}^*(\hat{\mathbf{b}}(\sigma))$ that generally differs from the optimal Nash equilibrium under perfect information. The resulting social welfare is
\begin{equation}\label{eq:noisy_opp}
	\xi\subscript{OPP}(\sigma) = \mathbb{E}_{\boldsymbol\varepsilon}\!\left[\sum_{i \in \mathcal{I}} \left(P_i(x_i^*) - \theta_i A(\mathbf{x}^*)\right) - C(\mathbf{x}^*)\right],
\end{equation}
where the expectation is over the noise realization and $\mathbf{x}^* = \mathbf{x}^*(\hat{\mathbf{b}}(\sigma))$.

Since \ac{iiq} does not use any parameter estimates, its welfare $\xi\subscript{IIQ}$ is independent of $\sigma$. We therefore have
\begin{equation}\label{eq:robustness_crossover}
	\xi\subscript{OPP}(0) = \xi\subscript{OPP} \geqslant \xi\subscript{IIQ}, \qquad \lim_{\sigma \to \infty} \xi\subscript{OPP}(\sigma) \to \xi\subscript{random},
\end{equation}
where $\xi\subscript{random}$ denotes the welfare under essentially random pricing. Since $\xi\subscript{OPP}(\sigma)$ is continuous and decreasing in $\sigma$, there exists a crossover point $\sigma^*$ such that
\begin{equation}\label{eq:crossover}
	\xi\subscript{IIQ} > \xi\subscript{OPP}(\sigma) \quad \text{for all } \sigma > \sigma^*.
\end{equation}
That is, when the server's parameter estimation error exceeds $\sigma^*$, the information-free quotation mechanism \emph{outperforms} the nominally optimal personalized pricing. We evaluate $\sigma^*$ numerically in Sec.~\ref{sec:evaluation}.

\section{Numerical Evaluation}\label{sec:evaluation}
\subsection{Simulation Setup}\label{subsec:setup}
We evaluate the proposed mechanism through numerical simulations. The default parameter settings are shown in Table~\ref{tab:setup}, following the specifications of \cite{CC2024price}. For the accuracy-aware model, we additionally set $\theta_i \sim \mathcal{U}(0, 5)$ and $k_i = 1$ (logarithmic privacy) unless stated otherwise. Welfare comparison, robustness, and fairness experiments are averaged over $N = 5000$ Monte Carlo runs, and sensitivity over $N = 500$ per configuration point, all with independent user population draws. The random seed is fixed at $42$ for reproducibility.

\begin{table}[!htbp]
	\centering
	\caption{Simulation Setup}
	\label{tab:setup}
	\begin{tabular}{>{\cellcolor{white}}m{0.2cm} | m{1.8cm} m{2.7cm} m{2.4cm}}
		\toprule[2px]
		&\textbf{Parameter}&\textbf{Value}&\textbf{Remark}\\
		\midrule[1px]

		\rowcolor{gray!20}
		&	$I$	&	$10$	&	Number of users\\
		&	$d_i, \forall i\in\mathcal{I}$	&	6000 &	Data amount per user\\

		\rowcolor{gray!20}
		&	$\lambda_i, \forall i\in\mathcal{I}$	&	$\sim\mathcal{U}(0.5,30)$ &	Privacy valuation\\
		&	$\theta_i, \forall i\in\mathcal{I}$	&	$\sim\mathcal{U}(0,5)$ &	Accuracy sensitivity\\

		\rowcolor{gray!20}
		&	$k_i, \forall i\in\mathcal{I}$	&	$1.0$ &	Privacy elasticity\\
		&	$\Delta d$	&	$1$	 &	Unit data amount\\

		\rowcolor{gray!20}
		&	$B^0$	&	$0.001$	 &	Initial price\\
		&	$\Delta B$	&	$0.001$	 &	Price step\\

		\rowcolor{gray!20}
		\multirow{-9}{*}{\rotatebox{90}{\textbf{System}}}&	$N$	&	$500$--$5000$ & Monte Carlo runs\\
		\midrule[1px]
		&	$[a, A_1,A_2,A_3]$	&	$[e, 0.1,3.33\times 10^{-5},0]$	&	Accuracy parameters\\

		\rowcolor{gray!20}
		&	$T_0$	&	$2.85\times 10^{-4}$	&	Time factor\\
		\multirow{-3}{*}{\rotatebox{90}{\textbf{Cost}}} &	$[\alpha,\beta]$	&	[1500,1]	&	Weight factors\\

		\bottomrule[2px]
	\end{tabular}
\end{table}

\subsection{Oversupply Handling Strategies}
To compare the performance of the four oversupply-handling strategies proposed in Sec.~\ref{subsec:oversupply}, we carried out Monte Carlo tests. For each run, we measured the server's payoff:
\begin{equation}
	\xi\subscript{s}=C^T-C(0)-\sum\limits_{t=0}^T B^t,
\end{equation}
where $T$ is the index of round when the procedure terminates, the users' total payoff
\begin{equation}
	\xi\subscript{u}=\sum\limits_{i=1}^I \left[\sum\limits_{t=1}^T B^{t-1} (y_i^t-y_i^{t-1})+U_i\left(y_i^T\right)-U_i(0)\right],
\end{equation}
and the social welfare
\begin{equation}
	\xi=\xi\subscript{s}+\xi\subscript{u}.
\end{equation}
The simulation results are shown in Table~\ref{tab:oversupply}. As we have expected, the selection of oversupply-handling strategy does not affect the server's payoff at all (difference within $0.2\%$), while the \emph{minor-seller-first} strategy outperforms the other three regarding the users' payoff and social welfare by a very slight superiority (within $2\%$). This is not surprising as the oversupply handling is triggered only in the last quotation round, which takes only a small portion of the total quotation procedure.
\begin{table}[!htbp]
	\centering
	\caption{Comparing the oversupply-handling strategies}
	\begin{tabular}{l|c|>{\columncolor{gray!20}}c|c|c}
			\toprule[2px]
			&\textbf{Major-first}&\textbf{Minor-first}&\textbf{Prop.}&\textbf{Rand.}\\
			\midrule[1px]
			\textbf{Server's payoff} & 684.8 & 683.6 & 684.1 & 684.9 \\
			\textbf{Users' payoff} & 1133.5 & \textbf{1158.2} & 1150.2 & 1134.0\\
			\textbf{Social welfare} & 1818.4 & \textbf{1841.8} & 1834.3 & 1818.9 \\
			\bottomrule[2px]
		\end{tabular}
	\label{tab:oversupply}
\end{table}

\subsection{Baseline Methods}\label{subsec:baselines}
We compare our quotation mechanism against six baselines spanning three categories:
\begin{enumerate}
	\item \textbf{Centralized baselines:}
	\begin{itemize}
		\item \emph{\acf{opp}:} The server knows all user parameters $(\lambda_i, k_i, \theta_i)$ and solves the bilevel optimization of \cite{CC2025price} to set personalized prices $b_i$ for each user. This is the social-welfare-maximizing benchmark under complete information.
		\item \emph{\Ac{bsp}:} The solution of \cite{CC2024price}, where the server sets a single optimal uniform price.
	\end{itemize}
	\item \textbf{Decentralized quotation:}
	\begin{itemize}
		\item \emph{\acf{ciq}:} Users play the \ac{spne} of Theorem~\ref{thm:spne} under full strategic foresight.
		\item \emph{\acf{iiq}:} The practical mechanism of Sec.~\ref{sec:approach} where users follow myopic best-response.
	\end{itemize}
	\item \textbf{Boundary baselines:}
	\begin{itemize}
		\item \emph{\Ac{dnr}:} No redemption; the server keeps all data for free.
		\item \emph{\Ac{gdpr}:} All informed users redeem all data; the server bears full unlearning cost.
		\item \emph{FULL:} The server retains all data but compensates each user for the full privacy loss $P_i(d_i) - P_i(0)$.
	\end{itemize}
\end{enumerate}
Following \cite{CC2024price}, the \emph{informed ratio} $\rho \in [0,1]$ denotes the fraction of users aware of their redemption rights; uninformed users are treated under \ac{dnr} regardless of the mechanism used.

\subsection{Privacy-Only Welfare Comparison}\label{subsec:eval_privacy}
We first evaluate the privacy-only model ($\theta_i = 0, k_i = 1$), which recovers the setting of \cite{CC2024price}. Fig.~\ref{fig:welfare_privacy} shows the social welfare of all seven mechanisms as the informed ratio $\rho$ varies from $0\%$ to $100\%$.

\begin{figure}[!htbp]
	\centering
	\includegraphics[width=\linewidth]{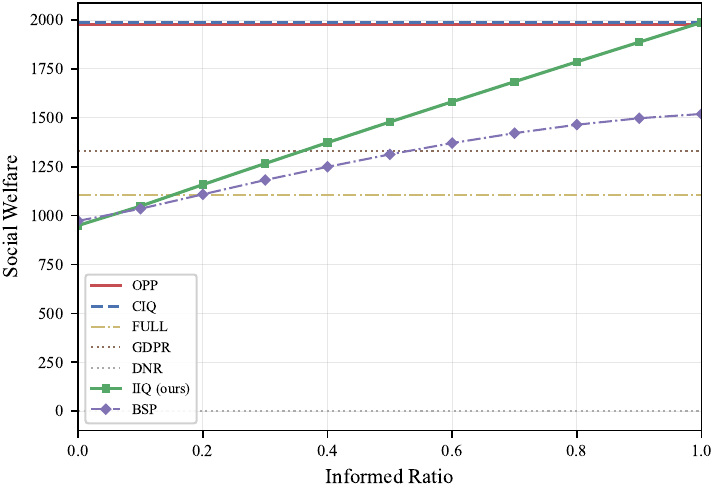}
	\caption{Privacy-only welfare comparison across mechanisms}
	\label{fig:welfare_privacy}
\end{figure}

\Ac{opp}, \ac{ciq}, and \ac{iiq} at full informed ratio ($\rho = 1$) achieve nearly identical welfare ($\approx 1976$, $\approx 1987$, and $\approx 1988$ respectively), confirming that in the privacy-only model, the Price of Ignorance is negligible---consistent with Sec.~\ref{subsec:price_of_ignorance}. \Ac{iiq} welfare increases monotonically with $\rho$, from $\approx 950$ at $\rho = 0$ to $\approx 1988$ at $\rho = 1$. \Ac{bsp} achieves $\approx 1520$ at full information, below \ac{iiq} but above \ac{gdpr}. The \ac{gdpr} baseline ($\approx 1329$) exceeds FULL ($\approx 1106$), since FULL must compensate all privacy loss.

\subsection{Accuracy-Aware Welfare Comparison}\label{subsec:eval_accuracy}
We now set $\theta_i \sim \mathcal{U}(0, 5)$ to activate the accuracy-aware model. Fig.~\ref{fig:welfare_accuracy} shows the welfare comparison.

\begin{figure}[!htbp]
	\centering
	\includegraphics[width=\linewidth]{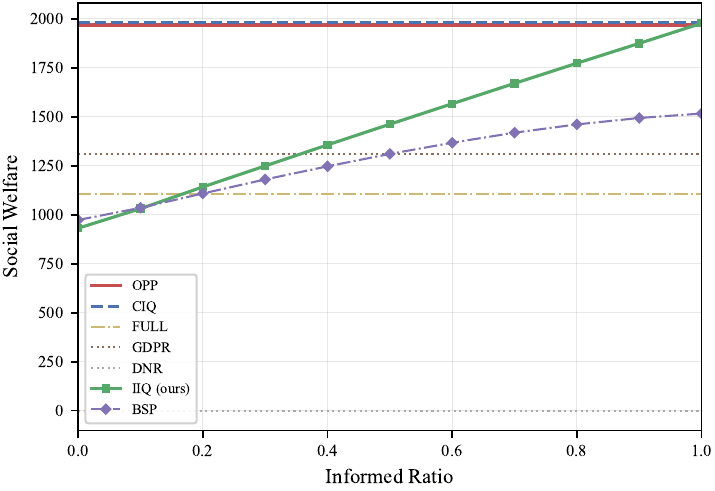}
	\caption{Accuracy-aware welfare comparison ($\theta_i \sim \mathcal{U}(0,5)$). The accuracy externality reduces welfare across all mechanisms.}
	\label{fig:welfare_accuracy}
\end{figure}

Under the accuracy-aware model, \ac{opp}, \ac{ciq}, and \ac{iiq} ($\rho = 1$) again achieve comparable welfare ($\approx 1968$, $\approx 1981$, and $\approx 1977$ respectively). The near-equivalence of the three regimes persists, consistent with the small Price of Ignorance predicted in Sec.~\ref{sec:efficiency}. Accuracy-aware users who are sensitive to model degradation retain more data, which reduces unlearning costs but also limits the scope for compensation-driven supply.

This near-equivalence does \emph{not} mean the three regimes are interchangeable. \Ac{opp} requires the server to know every user's private parameters $(\lambda_i, k_i, \theta_i)$, which is infeasible under data protection regulations and information asymmetry. \Ac{ciq} assumes users possess complete information about all others and play perfectly strategic equilibria---an equally unrealistic assumption. \Ac{iiq} is the \emph{only implementable} mechanism among the three: it requires no private parameter knowledge on either side and achieves $\geqslant 99\%$ of the theoretical optimum. The welfare gap in Figs.~\ref{fig:welfare_privacy}--\ref{fig:welfare_accuracy} is therefore not a shortcoming of \ac{iiq} but rather confirms that the information-free mechanism incurs virtually no efficiency loss relative to its information-intensive theoretical benchmarks.

\textbf{Under-supply vs.\ over-supply regimes: }
The results above reflect an \emph{under-supply} regime: the server's demand exceeds aggregate supply at the equilibrium price (demand fulfillment $\approx 66\%$). Proposition~\ref{prop:welfare_ordering} predicts that the welfare gap between \ac{opp} and \ac{iiq} widens in the \emph{over-supply} regime, where supply exceeds demand and the server must ration purchases. To verify this, we increase $\alpha = 10~000$ and $T_0 = 0.05$ so that the server's unlearning cost savings are large but time costs limit demand to $\approx 80\%$ of total data. Table~\ref{tab:supply_regimes} confirms the prediction: under over-supply, \ac{iiq} achieves only $61\%$ of \ac{opp} welfare because its uniform ascending price misallocates purchases across heterogeneous users, whereas \ac{opp}'s personalized pricing and \ac{ciq}'s endogenous ordering efficiently ration demand.

\begin{table}[!htbp]
\centering
\caption{Social welfare across supply regimes (accuracy-aware).}
\label{tab:supply_regimes}
\begin{tabular}{m{3.2cm} cccc}
	\toprule[2px]
	\textbf{Regime} & \textbf{OPP} & \textbf{CIQ} & \textbf{IIQ} & \textbf{IIQ/OPP} \\
	\midrule[1px]
	\rowcolor{gray!20}
	Under-supply (fill $\approx 66\%$) & 1968 & 1981 & 1977 & 1.00 \\
	Over-supply (fill $\approx 89\%$)  & 7581 & 7619 & 4604 & 0.61 \\
	\bottomrule[2px]
\end{tabular}
\end{table}

\Ac{iiq} is therefore most effective when the server's budget is the binding constraint (under-supply), which is the typical regime for machine unlearning where retraining costs are moderate and privacy valuations are high. In the over-supply regime, the server benefits from information about user heterogeneity to prioritize purchases.

\subsection{Robustness Under Parameter Estimation Noise}\label{subsec:eval_robustness}
To validate the robustness advantage predicted in Sec.~\ref{subsec:robustness}, we measure \ac{opp} welfare under multiplicative noise $\sigma \in \{0, 0.05, 0.1, 0.2, 0.5, 1.0, 2.0, 5.0\}$ on the server's estimates of $\lambda_i$ and $\theta_i$, as in Eq.~\eqref{eq:noisy_params}.

\begin{figure}[!htbp]
	\centering
	\includegraphics[width=\linewidth]{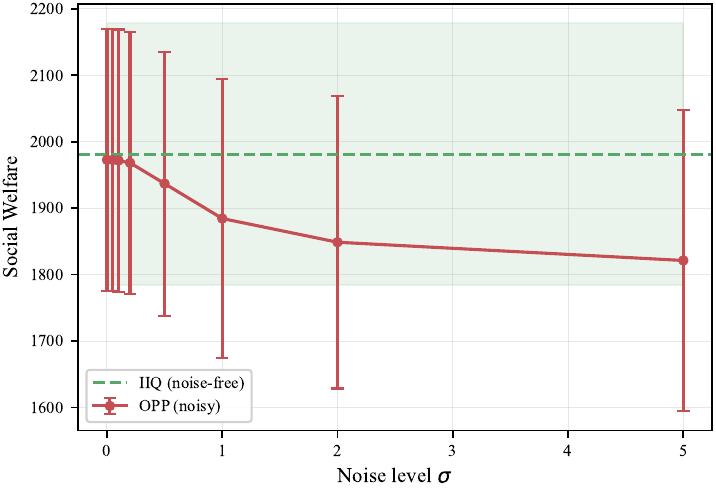}
	\caption{Robustness: \ac{opp} welfare degrades with estimation noise $\sigma$, while \ac{iiq} remains constant (green dashed line, with $\pm 1\sigma$ band).}
	\label{fig:robustness}
\end{figure}

Fig.~\ref{fig:robustness} shows the result. \Ac{opp} and \ac{iiq} achieve comparable welfare at $\sigma = 0$ ($\approx 1973$ and $\approx 1981$ respectively, within one standard error), but \ac{opp} welfare degrades monotonically to $\approx 1821$ at $\sigma = 5$ while \ac{iiq} remains constant. The welfare gap grows with noise: \ac{opp} loses $\approx 8\%$ of its welfare at $\sigma = 5$, whereas \ac{iiq} is entirely unaffected.

The crossover $\sigma^*$ predicted by Eq.~\eqref{eq:crossover} is not observed within the tested range because \ac{iiq} welfare is already competitive with \ac{opp} at $\sigma = 0$---a consequence of the near-zero Price of Ignorance. In effect, \ac{iiq} provides a guaranteed \emph{noise-independent} welfare level. In practice, two additional factors favor \ac{iiq}: (i)~the server avoids the operational cost of estimating user parameters, and (ii)~the information asymmetry is inherent---there is no reliable mechanism for the server to learn $(\lambda_i, \theta_i)$ without user cooperation. \Ac{iiq} circumvents this challenge entirely.

This strengthens the title's claim: since \emph{any} real deployment of \ac{opp} must operate at some $\sigma > 0$, the effective Price of Ignorance is even smaller than what Figs.~\ref{fig:welfare_privacy}--\ref{fig:welfare_accuracy} suggest (which compare against the idealized $\sigma = 0$). At moderate estimation noise ($\sigma \geqslant 0.5$), \ac{opp}'s realized welfare is already within \ac{iiq}'s confidence band, and at high noise ($\sigma \geqslant 2$), \ac{iiq} strictly dominates \ac{opp}---while requiring no estimation effort at all.

\subsection{Sensitivity Analysis}\label{subsec:eval_sensitivity}
We perform one-at-a-time sweeps of six key parameters to assess the robustness of our findings. For each sweep, we fix all other parameters at their default values and vary the target parameter across a range.

\begin{figure*}[!htbp]
	\centering
	\includegraphics[width=\linewidth]{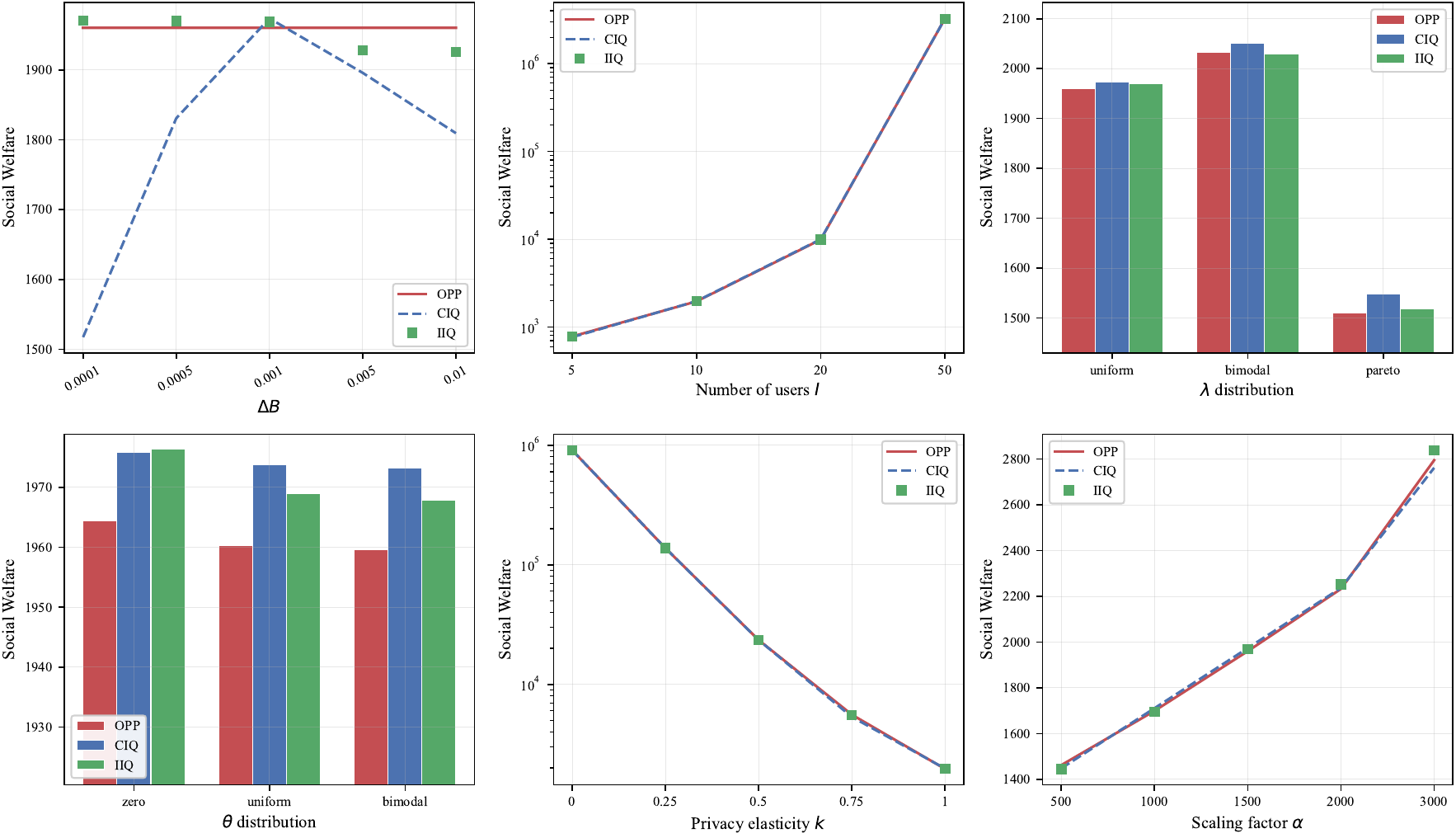}
	\caption{Sensitivity of social welfare to key parameters. Each panel sweeps one parameter while holding others at defaults.}
	\label{fig:sensitivity}
\end{figure*}

Fig.~\ref{fig:sensitivity} presents a six-panel summary:
\begin{itemize}
	\item \textbf{Price step $\Delta B$}: Smaller $\Delta B$ improves all mechanisms by enabling finer price discovery. The near-equivalence of \ac{opp}, \ac{ciq}, and \ac{iiq} is preserved across all values.
	\item \textbf{Number of users $I$}: Welfare scales with $I$ for all mechanisms. All three mechanisms remain tightly clustered even at $I = 50$.
	\item \textbf{Privacy valuation distribution}: Bimodal and Pareto distributions shift welfare levels but preserve the near-equivalence across mechanisms.
	\item \textbf{Accuracy sensitivity distribution}: Introducing non-zero $\theta_i$ has negligible impact on the gap between mechanisms, consistent with the dominance of the privacy term $P_i(x_i)$ over the accuracy externality $\theta_i A(x)$ at our parameter scale.
	\item \textbf{Privacy elasticity $k$}: Lower $k$ (more concave privacy function) dramatically increases welfare by reducing the marginal cost of selling. All three mechanisms respond similarly, preserving near-equivalence.
	\item \textbf{Scaling factor $\alpha$}: Higher $\alpha$ increases the server's cost savings from data retention, driving up welfare for all mechanisms proportionally.
\end{itemize}

\subsection{Fairness Analysis}\label{subsec:eval_fairness}
Beyond efficiency, we evaluate the fairness of payoff distribution across users using Jain's fairness index $\mathcal{J} = (\sum w_i)^2 / (I \sum w_i^2)$, the coefficient of variation, and the min-max payoff ratio.

\begin{figure*}[!htbp]
	\centering
	\includegraphics[width=\linewidth]{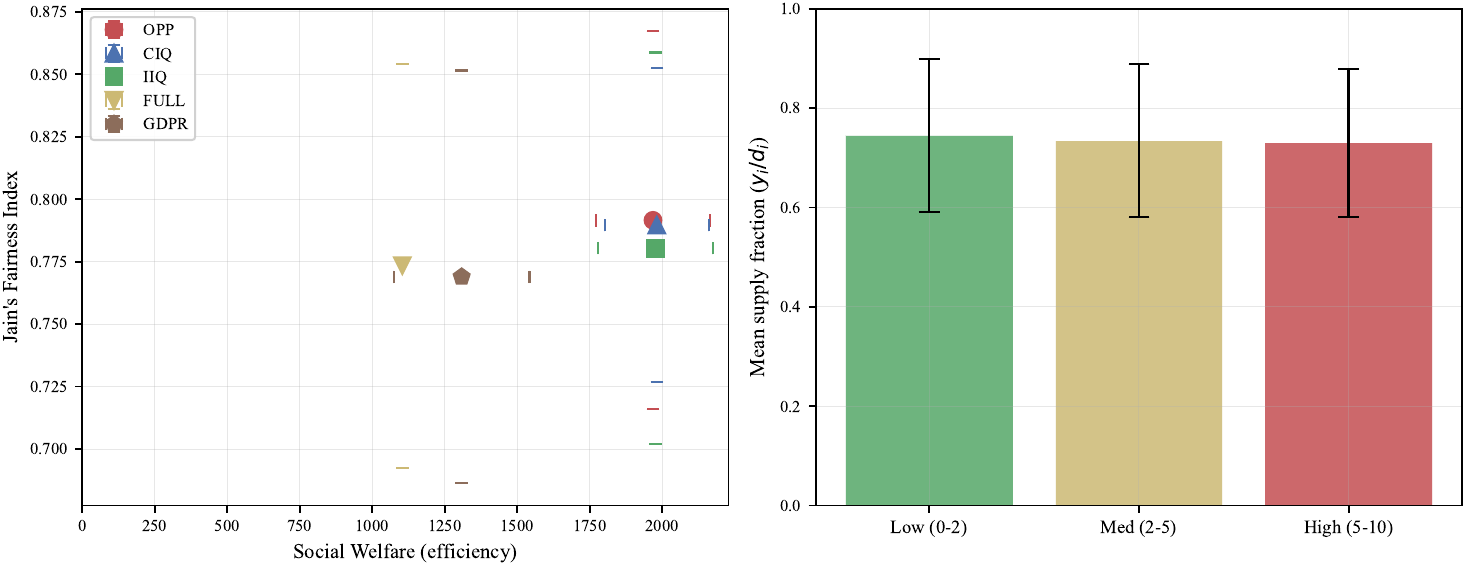}
	\caption{Left: Fairness-efficiency tradeoff across mechanisms. Right: Free-rider characterization---mean supply fraction by accuracy sensitivity level.}
	\label{fig:fairness}
\end{figure*}

Fig.~\ref{fig:fairness} (left) plots each mechanism in the social welfare vs.\ Jain's index space. \Ac{opp}, \ac{ciq}, and \ac{iiq} cluster together at high welfare ($\approx 1968$--$1981$) with comparable Jain's indices ($\mathcal{J} \approx 0.78$--$0.79$). \Ac{iiq} exhibits a slightly lower $\mathcal{J}$ than \ac{opp} and \ac{ciq} because its uniform ascending price does not adapt to individual heterogeneity: users with low $\lambda_i$ sell early at low prices and receive smaller payoffs, while high-$\lambda_i$ users retain more data and accumulate larger privacy utility. \Ac{opp}'s personalized prices and \ac{ciq}'s endogenous ordering partially compensate for this heterogeneity, yielding marginally tighter payoff distributions. \Ac{gdpr} and FULL cluster at lower welfare with lower fairness due to the asymmetric impact of accuracy loss across users with heterogeneous $\theta_i$.

Fig.~\ref{fig:fairness} (right) characterizes the free-rider effect. Under our default parameters, the mean supply fraction varies only slightly across $\theta_i$ bins (low: $0.75$, medium: $0.74$, high: $0.73$), suggesting that the accuracy externality has a limited impact on individual supply decisions when $\theta_i$ is moderate relative to $\lambda_i$. This is because the accuracy term $\theta_i A(x)$ is small compared to the privacy term $P_i(x_i)$ at our parameter scale. The free-rider effect becomes more pronounced when $\theta_i$ is large relative to $\lambda_i$.

\subsection{Convergence and Scalability}\label{subsec:eval_convergence}
Finally, we verify the convergence properties of the quotation mechanism.

\begin{figure}[!htbp]
	\centering
	\includegraphics[width=.85\linewidth]{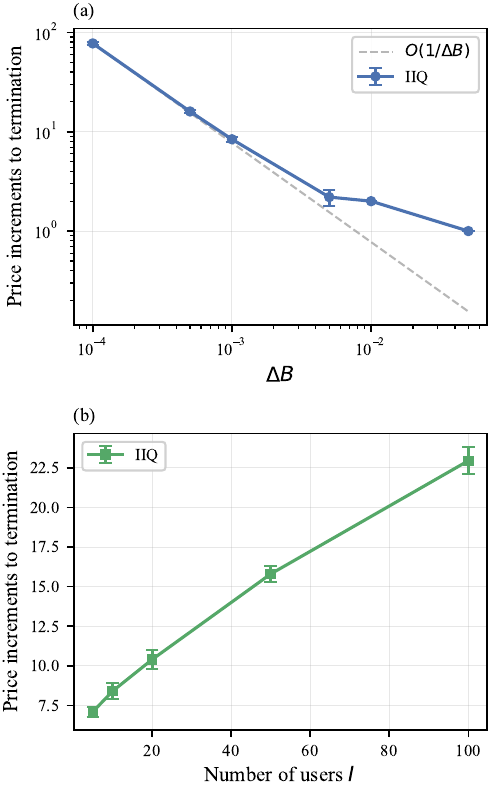}
	\caption{Convergence analysis. (a)~Price increments to termination vs.\ $\Delta B$ (log-log; dashed: $O(1/\Delta B)$ reference). (b)~Price increments vs.\ number of users $I$.}
	\label{fig:convergence}
\end{figure}

Fig.~\ref{fig:convergence} (left) confirms the $O(1/\Delta B)$ convergence rate: reducing $\Delta B$ from $0.05$ to $0.0001$ increases the number of price increments from $1$ to $\approx 78$. Each increment corresponds to one broadcast--collect cycle in Algorithm~\ref{alg:quotation}, where the server raises the price by $\Delta B$, users independently respond with supply decisions, and the server purchases. The log-log plot closely tracks the $O(1/\Delta B)$ reference line, validating the theoretical bound.

The right panel shows that the number of price increments grows with $I$ (from $7$ at $I = 5$ to $23$ at $I = 100$). Termination occurs when the server's marginal cost of raising the price exceeds its marginal cost savings---i.e., the server's surplus is exhausted. In all cases, demand is only partially filled at termination ($52\%$ at $I = 5$ to $89\%$ at $I = 100$), confirming that the binding constraint is the price cap, not demand saturation. The growth in rounds with $I$ arises because the total data pool $D = \sum d_i$ scales with $I$, which exponentially increases the unlearning cost $C(0) = \alpha A_1 e^{A_2 D}$ and hence the server's willingness to pay---the equilibrium termination price is higher, requiring more increments to reach from $B_0 = 0$. This suggests that an adaptive starting price $B_0(I)$ or geometric price increments could reduce the $I$-dependence (see Section~\ref{sec:discussion}).

\subsection{Summary: The Price of Ignorance Is Near Zero}\label{subsec:eval_summary}
\Ac{iiq} is the \emph{only} regime among \ac{opp}, \ac{ciq}, and \ac{iiq} that is deployable without private user information, yet it achieves welfare within $1\%$ of both information-intensive benchmarks. The Price of Ignorance---the welfare sacrificed by knowing nothing about users---is negligible. Table~\ref{tab:iiq_summary} consolidates the comparison.

\begin{table}[!htbp]
\centering
\caption{Practical comparison of the three regimes.}
\label{tab:iiq_summary}
\begin{tabular}{m{3.8cm} ccc}
	\toprule[2px]
	\textbf{Property} & \textbf{OPP} & \textbf{CIQ} & \textbf{IIQ} \\
	\midrule[1px]
	\rowcolor{gray!20}
	Private params.\ required by server & Yes & No & No \\
	Strategic reasoning required by users & No & Yes & No \\
	\rowcolor{gray!20}
	Implementable in practice & No & No & \textbf{Yes} \\
	Welfare at $\rho = 1$ (acc.-aware) & 1968 & 1981 & 1977 \\
	\rowcolor{gray!20}
	Robust to estimation noise & No & N/A & \textbf{Yes} \\
	Ex post individually rational & Yes & Yes & \textbf{Yes} \\
	\bottomrule[2px]
\end{tabular}
\end{table}

\Ac{opp} and \ac{ciq} serve as \emph{theoretical upper bounds} that quantify the Price of Ignorance. The answer---less than $1\%$ in the under-supply regime that characterizes typical machine unlearning deployments---validates \ac{iiq} as the mechanism of choice. The only scenario where \ac{iiq} underperforms substantially is the over-supply regime (Table~\ref{tab:supply_regimes}), where the server benefits from knowing user heterogeneity to prioritize purchases; this regime can be mitigated by operational design choices such as adaptive price scheduling (Sec.~\ref{sec:discussion}).

\section{Discussion}\label{sec:discussion}

\subsection{Voluntary Participation and Strategic Properties}

The mechanism satisfies \emph{ex post individual rationality}: at every round and for every realization of other users' actions, each user can guarantee non-negative payoff from that round by simply choosing not to sell. No user is ever forced to participate or penalized for abstaining.

Regarding supply manipulation, the price path $B^t = B^0 + t \cdot \Delta B$ is exogenous and independent of individual supply decisions. A price-taking user (whose individual supply does not influence aggregate supply significantly) therefore cannot affect the price trajectory or the termination time, making the myopic best-response strategy individually optimal. Under complete information with strategic users, Theorem~\ref{thm:spne} characterizes the unique \ac{spne}, from which no unilateral deviation is profitable by construction.


\subsection{Per-User Regret}
The cost of incomplete information can also be measured at the individual level through the \emph{ex post regret} of each user. Given the realized outcome under \ac{iiq}, define user $i$'s regret as
\begin{equation}\label{eq:regret}
	R_i \triangleq W_i^*(\mathbf{y}_{-i}^{\text{IIQ}}) - W_i^{\text{IIQ}},
\end{equation}
where $W_i^*(\mathbf{y}_{-i}^{\text{IIQ}})$ is the payoff user~$i$ would have achieved by playing the ex post best response, given full knowledge of others' realized actions and the termination price $B^T$. The aggregate regret $\sum_i R_i$ provides a per-user decomposition of the Price of Ignorance, identifying which users are most disadvantaged by myopic play. Users who begin selling early at low prices when they could have waited (or vice versa) exhibit higher regret. This metric is readily computable from our simulation outputs and offers a finer-grained view of the information gap than the aggregate welfare comparison alone.

\subsection{Public Good Structure and Free-Riding}
Data retention in our model has the structure of a \emph{voluntary contribution game}~\cite{BBV1986public}: model accuracy $1 - A(\sum_j x_j)$ is a public good that benefits all users, and each user's decision to retain data (rather than redeem it) constitutes a contribution to this public good. The free-riding incentive~\cite{Andreoni1988freeride} is to redeem one's own data---gaining privacy utility---while hoping that other users retain enough data to maintain model accuracy.
In this framing, high-$\theta_i$ users are natural \emph{contributors}: they value accuracy more, so the marginal benefit of others' retention is higher, and they are more willing to retain data themselves. Low-$\theta_i$ users are potential \emph{free-riders}: they are less affected by accuracy degradation and thus more inclined to redeem. Our fairness analysis confirms this pattern: the supply fraction varies only slightly across $\theta_i$ bins ($\approx$0.73--0.75), suggesting that the compensation mechanism $B$ acts as a Pigouvian subsidy that partially internalizes the accuracy externality and mitigates free-riding.

\subsection{Adaptive Price Scheduling}
Our convergence analysis (Section~\ref{subsec:eval_convergence}) reveals that the number of price increments grows with $I$ because the equilibrium termination price increases with total data volume. Since termination is driven by exhaustion of the server's surplus rather than demand fulfillment, the early low-price rounds where no user would sell are wasted. Two design improvements can reduce this $I$-dependence without sacrificing the information-free property: (i)~an \emph{adaptive starting price} $B_0 = f(D_\text{total})$ that skips the vacuous low-price region, and (ii)~\emph{geometric price increments} $B_{t+1} = (1 + r) B_t$ that yield $O(\log(B^*/B_0))$ rounds independent of the absolute price scale. Both modifications preserve the ascending-price structure and hence the incentive properties of Theorems~\ref{thm:single_period}--\ref{thm:spne}; the optimal choice of $B_0$ and $r$ as functions of publicly observable system parameters (e.g., total data volume $D$, number of users $I$) is an interesting direction for future work.

\subsection{Deployment Considerations}
For \acp{mno}, \ac{iiq} requires no user profiling or parameter estimation---the server need only broadcast a price and collect aggregate supply responses, simplifying compliance with data protection regulations that restrict processing of user behavioral data. The convergence rate scales as $O(1/\Delta B)$, giving operators a direct latency--precision tradeoff: a larger price step terminates faster but yields coarser welfare approximation. The mechanism also exhibits a natural resistance to collusion that stems from the public good structure of model accuracy. In a standard market, colluders benefit by jointly suppressing supply to raise prices. Here, colluding to withhold data retention would degrade model accuracy for \emph{all} colluders, since each accuracy-aware user ($\theta_i > 0$) benefits from others retaining data. This misalignment of collusive incentives---each user wants \emph{others} to contribute more, not less---makes sustained collusion self-defeating.

\subsection{Limitations}
Our model assumes that users are myopic under incomplete information and do not learn others' strategies across rounds. In practice, sophisticated users might infer aggregate supply from observable outcomes (e.g., model performance changes) and adjust behavior accordingly. We also model privacy utility as time-invariant, whereas privacy valuations may evolve as users become more aware of data risks. The accuracy loss function $A(\cdot)$ is treated as known to all parties; in practice, the server's accuracy model may itself be uncertain. Extending the framework to address these dynamic and informational aspects is left for future work.

\section{Conclusions}\label{sec:conclusion}
The Price of Ignorance in data retention pricing is near zero. An information-free ascending quotation---where the server knows nothing about users' privacy preferences or accuracy sensitivities---achieves $\geqslant 99\%$ of the welfare of personalized pricing that requires full parameter knowledge, while providing noise-robust guarantees and comparable fairness. The server sacrifices virtually no efficiency by not profiling its users.

To reach this conclusion, we extended the data redemption model with generalized privacy functions (elasticity $k_i$), accuracy-aware payoffs (sensitivity $\theta_i$), and heterogeneous data endowments. Under complete information, we proved that the ascending quotation admits a unique \ac{spne} characterized by single-period selling (Theorem~\ref{thm:single_period}) and backward induction (Theorem~\ref{thm:spne}). Our three-regime efficiency ordering---\ac{opp}~$\geqslant$~\ac{ciq}~$\geqslant$~\ac{iiq}---and its dependence on the oversupply-handling strategy (Proposition~\ref{prop:welfare_ordering}) formalizes exactly where and why ignorance costs welfare, and numerical evaluation across seven mechanisms and six parameter dimensions confirms that this cost is negligible in the under-supply regime typical of machine unlearning.

Future work will explore: \begin{enumerate*}[label=(\arabic*)]
	\item tighter analytical bounds on the Price of Ignorance ($\xi\subscript{CIQ} - \xi\subscript{IIQ}$) and per-user regret analysis;
	\item Bayesian belief models where users update their accuracy estimates over successive quotation rounds;
	\item extending the model to dynamic data accumulation where new data arrives between unlearning events; and
	\item strategic information sharing among users.
\end{enumerate*}


\bibliographystyle{IEEEtran}
\bibliography{references}

\end{document}

%% file: glossary.tex
\makeglossaries
\newacronym{ai}{AI}{artificial intelligence}
\newacronym{bsp}{BSP}{bi-stage pricing}
\newacronym{ccpa}{CCPA}{\emph{California Consumer Privacy Act}}
\newacronym{ciq}{CIQ}{Complete-Information Quotation}
\newacronym{dhr}{DHR}{decreasing hazard rate}
\newacronym{dnr}{DNR}{do-not-redeem}
\newacronym{foc}{FOC}{first-order condition}
\newacronym{gdpr}{GDPR}{\emph{General Data Protection Regulation}}
\newacronym{genai}{GenAI}{generative artificial intelligence}
\newacronym{iiq}{IIQ}{Incomplete-Information Quotation}
\newacronym{llm}{LLM}{large language model}
\newacronym{ml}{ML}{machine learning}
\newacronym{mno}{MNO}{mobile network operator}
\newacronym{opp}{OPP}{Optimal Personalized Pricing}
\newacronym{pipl}{PIPL}{\emph{Personal Information Protection Law}}
\newacronym{qos}{QoS}{quality of service}
\newacronym{spne}{SPNE}{subgame-perfect Nash equilibrium}